%% file: main.tex
\documentclass[%
 reprint,
 amsmath,amssymb, amsthm, amsfonts,
showkeys,
]{revtex4-1}

\usepackage{hyperref}
\usepackage{graphicx}
\usepackage{dcolumn}
\usepackage{bm}

\usepackage[utf8]{inputenc}
\usepackage{subfigure}
\usepackage[export]{adjustbox}
\usepackage[abs]{overpic}

\usepackage[T1]{fontenc}
\usepackage{mathptmx}

\usepackage[inline]{trackchanges}
\addeditor{Peter}
\addeditor{John}
\addeditor{Martin}
\addeditor{Christof}

\newtheorem{theorem}{Theorem}[section]
\newtheorem{lemma}[theorem]{Lemma}

\newtheorem{corollary}[theorem]{Corollary}

\newenvironment{equation-w-ref}[1][N/A]{
    \refstepcounter{equation}
    \label{#1-back-ref}
    \equation \tag{\mbox{\hyperref[#1]{$\Box$}} \theequation}}
    {\endequation}
    
\newenvironment{align-w-ref}[1][N/A]{
    \refstepcounter{equation}
    \label{#1-back-ref}
    \align \tag{\mbox{\hyperref[#1]{$\Box$}} \theequation}}
    {\endalign}

\newenvironment{theorem-w-back-ref}[1][N/A]{
    \renewcommand\thesection{E}
    \refstepcounter{theorem}
    \label{#1}
    \begin{trivlist}
        \item[\hskip \labelsep {\bfseries Theorem \hyperref[#1-back-ref]{\thetheorem}}]
        \itshape
    }{\end{trivlist}
    \renewcommand\thesection{A}
}

\newenvironment{lemma-w-back-ref}[1][N/A]{
    \renewcommand\thesection{E}
    \refstepcounter{theorem}
    \label{#1}
    \begin{trivlist}
        \item[\hskip \labelsep {\bfseries Lemma \hyperref[#1-back-ref]{\thetheorem}}]
        \itshape
    }{\end{trivlist}
    \renewcommand\thesection{A}
}

\newenvironment{corollary-w-back-ref}[1][N/A]{
    \renewcommand\thesection{E}
    \refstepcounter{theorem}
    \label{#1}
    \begin{trivlist}
        \item[\hskip \labelsep {\bfseries Corollary \hyperref[#1-back-ref]{\thetheorem}}]
        \itshape
    }{\end{trivlist}
    \renewcommand\thesection{A}
}

\newenvironment{proof}[1][Proof]{
    \begin{trivlist}
        \item[\hskip \labelsep \itshape {\bfseries #1}] \mbox{}
    }{\end{trivlist}
}

\newenvironment{example}[1][Example:]{
    \begin{trivlist}
        \item[\hskip \labelsep {\bfseries #1}]
    }{\end{trivlist}
}
\newenvironment{remark}[1][Remark:]{
    \begin{trivlist}
        \item[\hskip \labelsep \textit{#1}]
    }{\end{trivlist}
}

\newcounter{lemma-ref}

\newenvironment{lemma-ref}[1][N/A]{
    \refstepcounter{lemma-ref}
    \begin{trivlist}
        \item[\hskip \labelsep {\bfseries Lemma \ref{#1}}]
        \itshape
    }{\end{trivlist}
}

\newcounter{theorem-ref}

\newenvironment{theorem-ref}[1][N/A]{
    \refstepcounter{theorem-ref}
    \begin{trivlist}
        \item[\hskip \labelsep {\bfseries Theorem \ref{#1}}]
        \itshape
    }{\end{trivlist}
}
\newcounter{corollary-ref}

\newenvironment{corollary-ref}[1][N/A]{
    \refstepcounter{corollary-ref}
    \begin{trivlist}
        \item[\hskip \labelsep {\bfseries Corollary \ref{#1}}]
        \itshape
    }{\end{trivlist}
}

\newcommand{\niton}{\not\owns}

\newcommand{\binomempty}{\genfrac{}{}{0pt}{}}

\graphicspath{{images/}}

\begin{document}

\preprint{abs/1703.09030}

\title{Shift-Symmetric Configurations in Two-Dimensional Cellular Automata:\\ Irreversibility, Insolvability, and Enumeration}

\author{Peter Banda}
\email{peter.banda@uni.lu}
\affiliation{Luxembourg Centre For Systems Biomedicine\\University of Luxembourg, Esch-sur-Alzette, L-4362, Luxembourg \looseness=-1}

\author{John Caughman}
\email{caughman@pdx.edu}
\affiliation{Department of Mathematics and Statistics\\Portland State University, Portland, OR, 97201, USA}

\author{Martin Cenek}
\email{cenek@up.edu}
\affiliation{Shiley School of Engineering\\University of Portland, Portland, OR, 97203, USA}

\author{Christof Teuscher}
\email{teuscher@pdx.edu}
\affiliation{Department of Electrical and Computer Engineering\\Portland State University, Portland, OR, 97201, USA}

\date{June 26, 2019}

\begin{abstract}



The search for symmetry as an unusual yet profoundly appealing phenomenon, and the origin of regular, repeating configuration patterns have long been a central focus of complexity science and physics. To better grasp and understand symmetry of configurations in decentralized toroidal architectures, we employ group-theoretic methods, which allow us to identify and enumerate these inputs, and argue about irreversible system behaviors with undesired effects on many computational problems. The concept of so-called \textit{configuration shift-symmetry} is applied to two-\mbox{dimensional} cellular automata as an ideal model of computation. Regardless of the transition function, the results show the universal insolvability of crucial distributed tasks, such as leader election, pattern recognition, hashing, and encryption. By using compact enumeration formulas and bounding the number of shift-symmetric configurations for a given lattice size, we efficiently calculate the probability of a configuration being shift-symmetric for a uniform or density-uniform distribution. Further, we devise an algorithm detecting the presence of shift-symmetry in a configuration.

Given the resource constraints, the enumeration and probability formulas can directly help to lower the minimal expected error and provide recommendations for system's size and initialization. Besides cellular automata, the shift-symmetry analysis can be used to study the non-linear behavior in various synchronous rule-based systems that include inference engines, Boolean networks, neural networks, and systolic arrays. 

\end{abstract}



\maketitle

\begin{quotation}
Symmetry is a synonym for beauty and rarity, and generally perceived as something desired. In this paper we investigate an opposing side of symmetry and show how it can irreversibly \textit{corrupt} a computation, and restrict a system's dynamics and its potentiality. We demonstrate this fundamental phenomenon, which we call \textit{configuration shift-symmetry}, affecting many crucial distributed tasks on the simplest grid-like synchronous system of cellular automaton. We show how to count these symmetric inputs depending on a lattice size and its prime factorization, how likely they are encountered, and how to detect them.
\end{quotation}

\input{1-Introduction.tex}

\input{2-CellularAutomaton.tex}
\input{3-SymmetricConfigs.tex}

\input{4-EnumSymmetricConfigs.tex}
\input{5-EnumSymmetricConfigsForkActiveCells.tex}
\input{6-DetectSymmetricConfig.tex}
\input{7-Conclusion.tex}

\section*{Acknowledgments}
This material is based upon work supported by the National Science Foundation under grants \#1028120, \#1028378, \#1518833, and by the Defense Advanced Research Projects Agency (DARPA) under award \#HR0011-13-2-0015. The views expressed are those of the author(s) and do not reflect the official policy or position of the Department of Defense or the U.S. Government. Approved for Public Release, Distribution Unlimited.

\clearpage

\appendix
\input{8-Proofs.tex}

\clearpage

\input{9-Appendix.tex}

\clearpage

\section*{References}

\bibliography{common}
\bibliographystyle{aipauth4-1}

\end{document}

%% file: 1-Introduction.tex
\section{Introduction}
\label{sec:introduction}

The structure of the computational rules that result in regular, repeating system configurations has been studied by many, yet the question of how the natural and engineered system organize into symmetric structures is not completely known. To understand the role of symmetry of the starting configurations (the inputs), how they are processed (the machine), and produce the final configurations with desired properties (the outputs) we use a cellular automata (CA) as a simple distributed model of computation. First introduced by John von Neumann, CAs were instrumental in the exploration of logical requirements for machine self-replication and information processing in nature \cite{Neu66}. Despite having no central control and limited communication among the components, CAs are capable of universal computation and can exhibit various dynamical regimes \cite{Berlekamp1982,  Smith71, Wolfram1984}. As one of the structurally simplest distributed systems, CAs have become a fundamental model for studying complexity in its purist form \cite{CMD03,Wol86}. Subsequently, CAs have been successfully employed in numerous research fields and applications, such as modeling artificial life \cite{Langton90}, physical equations \cite{Fredkin82,Vichniac84}, and social and biological simulations \cite{Ermentrout93,Sales1997,Kal2015,Santos2001}.

The CA input configurations define a language that is processed by the machine. Exploring the structural symmetries of the input language not only translates to an efficient machine implementation, but allows us to argue about a problem insolvability and the irreversibility of computation.


In this paper, we explore the concept of shift-symmetry and revisit a well-known fact that any standard CA maintains a configuration shift-symmetry due to uniformity and synchronicity of cells. We show that once a system reaches a symmetric, i.e., spatially regular configuration, the computation will never revert from this attractor and will fail to solve all problems that require asymmetric solutions. As a result, the number of symmetries of the dynamical system is never decreasing. When a configuration slips to a symmetric, repeating pattern the configuration space of the CA irreversibly folds, causing a permanent regime ``shift." Consequently, a non-symmetric solution cannot be reached from a shift-symmetric configuration. A more general implication is that a configuration is unreachable (even if symmetric) if a source configuration has a symmetry not contained in the target. Non-symmetric tasks, such as leader election or pattern recognition, i.e., tasks expecting a final configuration to be non-symmetric, are therefore principally insolvable, since for any lattice size there always exist input configurations that are symmetric. As a hypothesis we also briefly discuss the eventual gradual increase of system's symmetries at the end of this paper, however, without any strong claims or proofs attached.

Using basic results from group theory and elementary combinatorics, we develop three progressively more efficient enumeration techniques based on mutually independent generators to answer the question of how many potential shift-symmetric configurations there are in any given two-dimensional CA lattice. As a side product, we demonstrate that the shift-symmetry is closely linked to prime factorization. We introduce and prove lower and upper bounds for the number of shift-symmetric configurations, where the lower bound (local minima) is tight and reached only for prime lattice sizes. We enumerate shift-symmetric configurations for a given lattice size and number of active cells.

Finally, we derive a formula and bounds for the probability of selecting shift-symmetric configuration randomly generated from a uniform or density-uniform distribution. We develop a shift-symmetry detection algorithm and prove its worst and average-case time complexities.

\subsection{Applications}
\label{subsec:applications}

All the formulas and proofs presented in this paper assume a two-dimensional CA with any number of states, and arbitrary uniform transition and neighborhood functions, which makes our results widely applicable.

Knowing the number of shift-symmetric configurations, we can directly determine the probability of selecting a shift-symmetric configuration by chance. This probability then equals an error lower bound or expected insolvability for any non-symmetric task. As we show, the insolvability caused by shift-symmetry rapidly decreases asymptotically with the lattice size for a uniform distribution. For instance, the probability is $0.5$ for a $2\times2$ lattice, but drops to around $2.7\times10^{-15}$ for a $10\times10$ lattice. Since the number of shift-symmetric configurations heavily depends on the prime factorization of the lattice size, the probability function is non-monotonously decreasing.  
To minimize the occurrence of shift-symmetries for uniform distribution, we generally recommend using prime lattices, or at least avoiding even ones. On the other hand, the probability for a density-uniform distribution is quite high, regardless of primes; it is around $10^{-3}$, even for a $45\times45$ lattice.

The distribution error-size constraints have important consequences for designing robust and efficient computational procedures for many crucial distributed problems, such as leader election \cite{Smith71,Banda11}, pattern recognition \cite{Rosin2006}, edge detection \cite{Slatnia2007}, image translation \cite{ioannidis2014}, convex hull/minimum bounding rectangle \cite{Cenek11}, hashing or collision resolution for associative memory \cite{chowdhury1995}, encryption \cite{wang2013}, and random number generation \cite{tomassini2000}. For these tasks an expected final configuration, e.g., reproduction of a certain two-dimensional image, is frequently non-shift-symmetric, and therefore unreachable from a symmetric configuration. Alternatively, an expected configuration can be unreachable even if it is shift-symmetric, which occurs when the vector space of its generating vectors (shifts) do not contain all the shifts of an initial configuration.

Practical implications of these properties include performance degradation of systolic CPU arrays and nanoscale multicore systems \cite{Zhirnov2008}. Our results span to the hardware implementations of synchronous CAs with \mbox{FPGA}s, used, e.g., for traffic signals control \cite{kalogeropoulos2013cellular}, random number generation \cite{shackleford2002fpga}, and reaction-diffusion model \cite{ishimura2015fpga}; and spintronics, where computation is achieved by coupled oscillators \cite{boehme2013challenges,vodenicarevic2016synchronization}. Also, current efforts to implement two or three-dimensional cellular automata using DNA tiles \cite{gu2009dynamic,wei2012complex} and/or gel-separated compartments in so-called \textit{gellular automata} \cite{hagiya2014dna,kawamata2016discrete} may face problems related to configuration shift-symmetry if a synchronous update is considered.



\subsection{Related Work}
\label{subsec:related-work}

In their seminal work, Packard and Wolfram \cite{packard1985} identified the importance of symmetry and showed that the global properties of a CA emerge as a function of the transition function's reflective and rotational symmetries. The fundamental algebraic properties of additive and non-additive CAs were studied by Martin \textit{et al.} \cite{martin1984}, who demonstrated that in simple cases there is a connection between the global behavior and the structure of the configuration transitions. Wolz and deOliveira \cite{wolz2008} exploited the structure and symmetry in the transition table to design an efficient evolutionary algorithm that found the best results for the density classification and parity problems. Marquez-Pita \textit{et al.} \cite{Marques2008, Marques2011} used a brute-force approach to find similar input configurations that produce the same outputs. Their results are a compact transition function re-description schema that use wild-cards to represent the many-to-one computation rules on a majority problem. Bagnoli \textit{et al.} \cite{Bagnoli2012} explored different methods of master-slave synchronization and control of totalistic cellular automata. A number of computation-theoretic results for CA were summarized by Culik II \textit{et al.} \cite{CulikII1990}, who investigated CAs through the eyes of set theory and topology. The effect of symmetry on the complexity of Boolean functions was thoroughly researched by Babai \textit{et al.} \cite{babai1992}. Pippenger \cite{pippenger1994} studied translation functions capable of correcting CA configurations under a specific kind of symmetry: rotation (an isometry with a corner coordinate fixed). 

Besides the symmetry of \textit{transition functions} and the design of transition functions resulting in regular or synchronized patterns, a number of contributions to the theoretical CA literature have addressed the general structure and implications of shift-symmetric \textit{configurations} also called translation invariant, or simply periodic, as we do here. This problem has been studied primarily in the context of group theory \cite{ceccherini2010}, through a general approach using stabilizers, group actions, and Bernoulli shifts. In particular, the work by Castillo-Ramirez and Gadouleau \cite{castillo2016} for example, approaches the problem using M\"{o}bius inversion of the subgroup lattice. Our derivation differs from their work by leveraging the affordances of specifying in advance that our symmetries are restricted to shift-symmetries, i.e., the specific case of Cartesian powers of cyclic groups in two dimensions, and proceeding inductively, which allows us to derive stronger results for the subproblem of our interest. In particular, we provide more efficient and executable enumeration formulas in an algorithmic sense and a better lower bound for the number of aperiodic configurations. Note that Castillo-Ramirez and Gadouleau improved the bound found by Gao \textit{et al.} \cite{gao2016} 

Another recent article \cite{ethier2013} explores similar questions regarding the number of distinct binary configurations of toroidal arrays in the presence of rotational and reflection symmetries. For our purposes, the ratio of symmetric to non-symmetric configurations is of greater interest than a simple enumeration of the total. Accordingly, our work differs from theirs by our focus on enumerating how many of these configurations possess some nontrivial symmetry (and, additionally, we do not wish to be limited to the binary case of an alphabet of size 2).

The concept of symmetry in number theory has been applied to so-called tapestry design and periodic forests \cite{apsimon1970,miller1970}, which relates to CA configurations. However, the triangular topology and geometric branching differs from the discrete toroidal Cartesian topology typically used for CAs. 



One of our main motivations is the pioneering work of Angluin \cite{Angluin1980}, who noticed that a ring containing anonymous components (processors), which are all in the same state, will never break its homogeneous configuration and elect a leader. This intuitive observation is, in fact, a special case of the concept of configuration shift-symmetry for CAs. We will show that Angluin's homogeneous state, which corresponds to a configuration of all zeros or all ones in a binary CA, is the most symmetric configuration for a given lattice size.


The concept of shift-symmetry is related to the notion of regular domains in computational mechanics \cite{HC92,CH93,Rupe2018,Rupe2018b}. A shift-symmetric configuration is essentially a (global) regular domain spread to a full lattice. Although we cannot apply our results directly to regular domains at the level of sub-configurations, because we pay no attention to local symmetries and non-cyclic and non-regular borders, the number of possible shift-symmetric configurations gives at least an upper bound on the number of possible regular domains. 


In our previous work \cite{Banda2015} we proved that configuration shift-symmetry, along with loose-coupling of active cells, prevents a leader from being elected in a one-dimensional CA \cite{Banda11}. The \textit{leader election} problem, first introduced by Smith \cite{Smith71}, requires processors to reach a final configuration where exactly one processor is in a leader state (one) and all others are followers (zero). Leader election is representative of a problem class where the solution is an asymmetric, non-homogeneous, transitionally and rotationally invariant system configuration. A final fixed-point configuration is asymmetric, since it contains only one processor in a leader state. Clearly, leader election and symmetry are \textit{enemies}, and, in fact, leader election is often called symmetry-breaking. 

To enumerate shift-symmetric configurations for a one-dimensional case \cite{Banda2015} we employed only basic combinatorics. Here, in order to span to two dimensions, we extend our enumeration machinery to include some basic concepts from group theory and we rely heavily on the notion of independent generators. We show that the insolvability caused by configuration symmetry extends beyond leader election to a whole class of non-symmetric problems.


%% file: 2-CellularAutomaton.tex
\subsection{Model}
\label{sec:ca-definition}

By definition, a CA \cite{Cod68} consists of a lattice of $N$ components, called {\em cells},
and a {\em state set} $\Sigma$.
A state of the cell with index $i$ is denoted $s_{i} \in \Sigma$.
A {\em configuration} is then a sequence of cell states:
\begin{equation}
    \mathbf{s} = (s_{0}, s_{1},\dots, s_{N - 1}).
\end{equation}
Given a topology for the lattice and the number of neighbors $b$, a {\em neighborhood} function $\eta: \mathbb{N}\times\Sigma^{N} \to \Sigma^{b}$ maps any pair $(i, \mathbf{s})$ to the $b$-tuple $\eta_i(\mathbf{s})$ of cells' states that are accessible (visible) to cell $i$ in configuration $\mathbf{s}$. Note that each cell is usually its own neighbor.

The transition rule \( \phi : \Sigma^{b} \to \Sigma \) is applied in parallel to each cell's neighborhood, resulting in the synchronous update of all of the cells' states \( s_{i}^{t+1} = \phi(\eta_{i}(\mathbf{s})^{t}) \).
The transition rule is represented either by a transition table, also called a look-up table, or a finite state transducer \cite{Hor00}. Here we focus exclusively on {\em uniform} CAs, where all cells share the same transition function. 
 The {\em global transition rule} \( \Phi: \Sigma^{N} \to
\Sigma^{N} \) is defined as the transition rule with the scope over all configurations: $\mathbf{s}^{t+1} = \Phi(\mathbf{s}^{t})$.

\begin{figure}[t!]
    \begin{center}
        \includegraphics[width=0.36\textwidth]{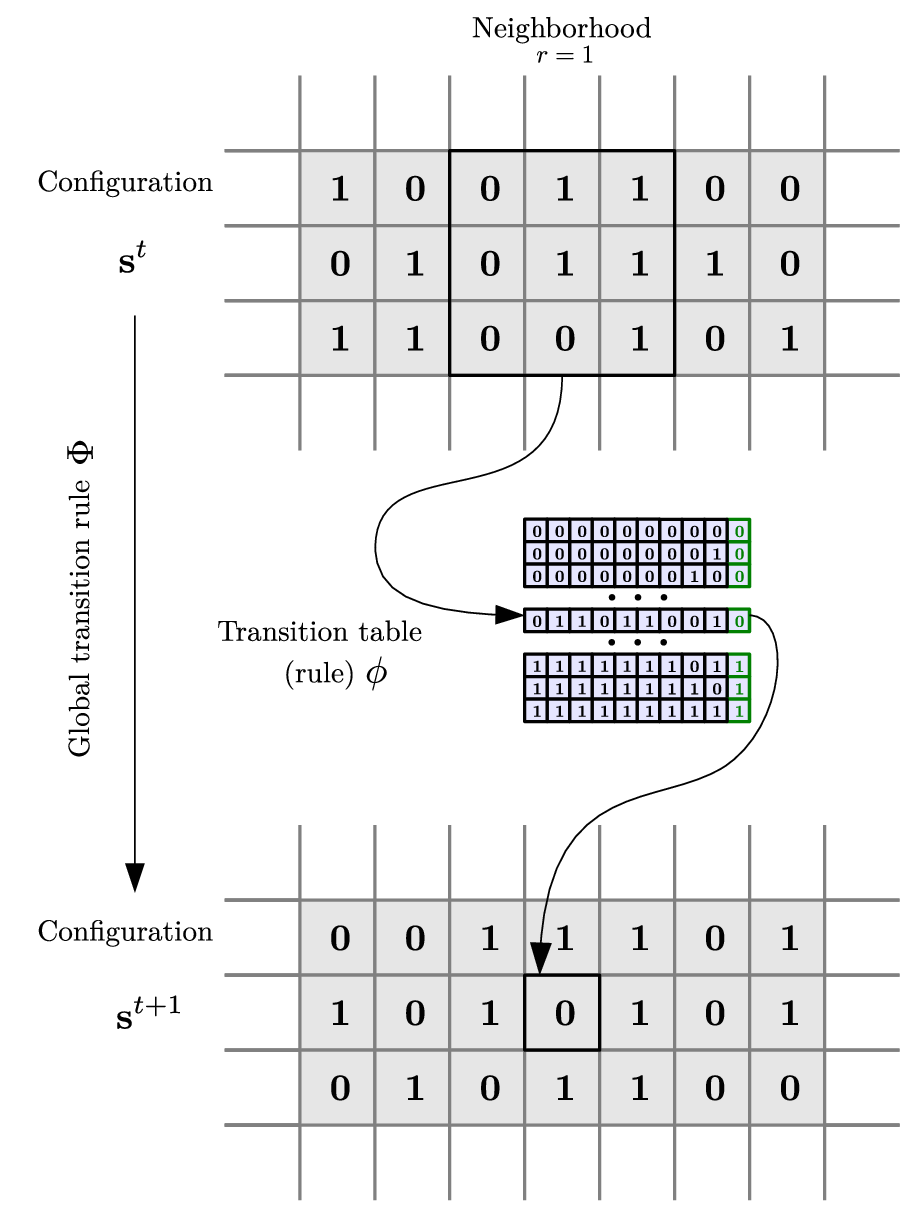}
     \end{center}
        \caption{Schematic of the configuration update for a binary two-dimensional CA, where the first nine bits in the transition table represent the flattened Moore neighborhood and the last bit is the output.}
        \label{fig:2d_ca-update}
\end{figure}

\begin{figure}[t!]
    \begin{center}
    \subfigure {
        \parbox{0.102\textwidth}{
            \includegraphics[width=0.102\textwidth,frame]{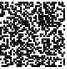} \\
            $t_{0}$
        }
    }
    \subfigure {
        \parbox{0.102\textwidth}{
            \includegraphics[width=0.102\textwidth,frame]{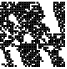} \\
            $t_{40}$
        }
    }
    \subfigure {
        \parbox{0.102\textwidth}{
            \includegraphics[width=0.102\textwidth,frame]{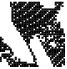}\\
            $t_{80}$
        }
    }
    \subfigure {
        \parbox{0.102\textwidth}{
            \includegraphics[width=0.102\textwidth,frame]{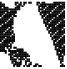}\\
            $t_{120}$
        }
    }
    \subfigure {
        \parbox{0.102\textwidth}{
            \includegraphics[width=0.102\textwidth,frame]{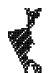}\\
            $t_{160}$
        }
    }
    \subfigure {
        \parbox{0.102\textwidth}{
            \includegraphics[width=0.102\textwidth,frame]{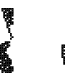}\\
            $t_{180}$
        }
    }
    \subfigure {
        \parbox{0.102\textwidth}{
            \includegraphics[width=0.102\textwidth,frame]{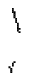}\\
            $t_{200}$
        }
    }
    \subfigure {
        \parbox{0.102\textwidth}{
            \includegraphics[width=0.102\textwidth,frame]{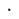}\\
            $t_{212}$
        }
    }
    \end{center}
    \caption{Example space-time diagrams of a leader-electing CA on lattice size $N = 40^2$ from \cite{Banda2014}. Figures show a CA computation starting with a random initial configuration (time $t_0$), followed by $7$ configuration snapshots. The CA reaches a final configuration with a single active cell (leader) at time $t_{212}$.}
    \label{fig:2dca}
\end{figure}

In this paper we analyze two-dimensional CAs, where cells are topologically organized on a two-dimensional grid with cyclic boundaries, i.e., we treat them as tori. 
The true power of our analysis is that it applies to  two-dimensional CAs with arbitrary neighborhood and transition functions. We rely only on their uniformity: each cell has the same neighborhood and transition function; and synchronous update, the attributes typically assumed for a standard CA.



Figure \ref{fig:2d_ca-update} shows the update mechanism for a two-dimensional binary CA with a Moore neighborhood, a square neighborhood with radius $r = 1$ containing $9$ cells. The dynamics of two-dimensional CAs are illustrated as a series of configuration snapshots, where an active cell is black and an inactive cell white (Figure \ref{fig:2dca}).

%% file: 3-SymmetricConfigs.tex
\section{Shift-Symmetric Configurations}
\label{sec:symmetric-configs}

As stated by Angluin \cite{Angluin1980}, the problem of reaching a ``center'' (i.e., leader) in homogeneous configurations is insolvable by any anonymous deterministic algorithm (including CAs). The CA uniformity can be embedded in its transition function, the deterministic update, synchronicity, topology, configuration, and cells' anonymity. Intuitively, a fully uniform system in terms of its structure, configuration, and computational mechanisms cannot produce any reasonable or complex dynamics. 

We show that Angluin's homogeneous configurations of $0^N$ and $1^N$ belong to a much larger class of so-called shift-symmetric configurations. In this section we formalize the concept of configuration shift-symmetry by employing vector translations and group theory. Figure \ref{fig:symmetric-ca-2d} depicts a CA computation on a two-dimensional shift-symmetric configuration. Compared to the one-dimensional case \cite{Banda2015}, two dimensions are more symmetry-potent.

\begin{figure}[b!]
    \begin{center}
        \includegraphics[width=0.23\textwidth]{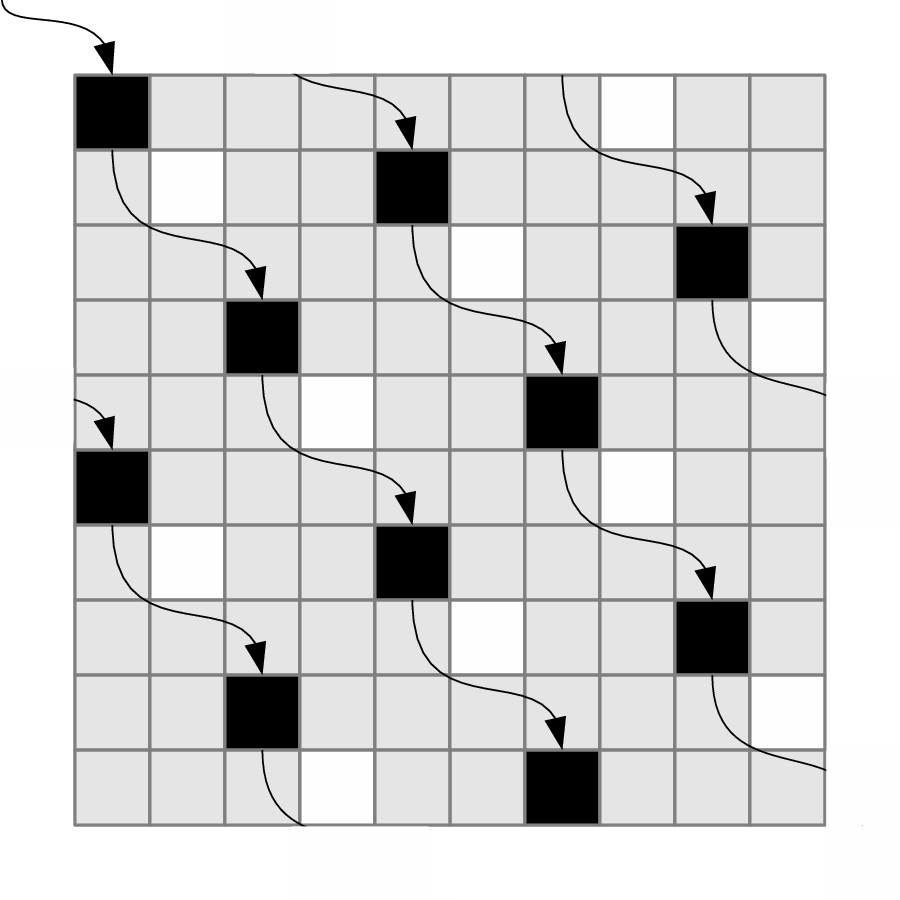}
     \end{center}
        \caption{Schematic of a shift-symmetric two-dimensional configuration generated by the vector ${\bf v} = (2,3)$ on $Z_{10 \times 10}$.}
        \label{fig:2d_symmetric-vector}
\end{figure}

It is important to mention that we deal with \textit{square} configurations only. Nevertheless, we suggest most of the lemmas and theorems could be extended to incorporate arbitrary rectangular shapes.
Also, the formulas and methodology to enumerate two-dimensional shift-symmetric configurations could be generalized to arbitrarily many dimensions. For consistency, however, we leave the rectangular as well as $n$-dimensional extensions for future consideration. Note that in order to improve the readability of the main text, all proofs and formally defined lemmas and theorems appear in Appendix \ref{appendix:proofs}. The non-trivial proofs from the appendix are referenced by the $\Box$ symbol.


First, we define a shift-symmetric (square) configuration by a given vector as shown in Figure \ref{fig:2d_symmetric-vector}. Formally, for a non-zero vector (pattern shift) ${\bf v} \in \mathbb{Z}_n \times \mathbb{Z}_n$ we denote by
\begin{equation}
S_{n \times n}({\bf v}) = \{\mathbf{s} \in \Sigma^{n \times n} \,|\, \forall {\bf u} \in \mathbb{Z}_n \times \mathbb{Z}_n : s_{{\bf u}} = s_{{\bf u} \oplus {\bf v}}\}
\end{equation}
the set of all {\bf shift-symmetric square configurations} of size $N = n^2$ relative to ${\bf v}$ over the alphabet $\Sigma$, where $\oplus$ denotes coordinate-wise addition on $\mathbb{Z}_n \times \mathbb{Z}_n$. 

\begin{figure}[t!]
    \begin{center}
    \subfigure {
        \parbox{0.14\textwidth}{
            \includegraphics[width=0.14\textwidth,frame]{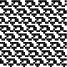} \\
            $t_0$
        }
    }
    \subfigure {
        \parbox{0.14\textwidth}{
            \includegraphics[width=0.14\textwidth,frame]{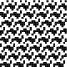} \\
            $t_1$
        }
    }
    \subfigure {
        \parbox{0.14\textwidth}{
            \includegraphics[width=0.14\textwidth,frame]{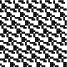} \\
            $t_2$
        }
    }
    \end{center}
    \caption{Space-time diagrams of CA computation on a two-dimensional binary shift-symmetric configuration showing a lattice at three consecutive time steps. Once reached, a shift-symmetry cannot be broken.}
    \label{fig:symmetric-ca-2d}
\end{figure}

Note that as opposed to our previous work \cite{Banda2014}, we renamed \textit{symmetric} to \textit{shift-symmetric} configurations to avoid confusion with reflective or rotational symmetries. These two symmetry types, unlike shift-symmetry, are not generally preserved by a transition function unless we impose certain ``symmetric'' properties on the transitions.

Since any translation by a non-zero vector ${\bf v}$ defines a configuration symmetry, we can study shift-symmetric configurations with the techniques of group theory. From now on, we will call such a non-zero vector ${\bf v} \in \mathbb{Z}_n \times \mathbb{Z}_n$ that we use for state translation a {\em generator} formalized as
\begin{equation}
S_{n \times n}({\bf v}) = \{\mathbf{s} \in \Sigma^{n \times n} \,|\, \forall {\bf u} \in \mathbb{Z}_n \times \mathbb{Z}_n \forall {\bf w} \in \langle {\bf v} \rangle : s_{{\bf u}} = s_{{\bf u} \oplus {\bf w}}\},
\end{equation}
where $\langle {\bf v} \rangle$ is the cyclic subgroup of  $\mathbb{Z}_n \times \mathbb{Z}_n$ generated by ${\bf v}$.

Trivially, for any non-zero ${\bf v} \in \mathbb{Z}_n \times \mathbb{Z}_n$
\begin{equation-w-ref}[lemma:symmetric-config-set-2d-l1_l2-size]
|S_{n \times n}({\bf v})| = |\Sigma|^{\frac{n^2}{|\langle {\bf v} \rangle|}}.
\end{equation-w-ref}









In the following text we bridge shift-symmetry, which is associated with configurations, i.e., static \textit{states}, with any uniform transition rule, which defines the \textit{dynamics} of CA. We show that shift-symmetry cannot be broken, thus it fundamentally restricts the reachable states and the potentiality of a transition rule. More formally, for a vector ${\bf v}$ and any uniform global transition rule $\Phi$
\begin{equation-w-ref}[theorem:symmetric-stays-symmetric-2d]
{\bf s} \in S_{n \times n} ({\bf v}) \implies \Phi(\mathbf{s}) \in S_{n \times n} ({\bf v}),
\end{equation-w-ref}



and so, by induction, a non-symmetric configuration ${\bf q} \notin S_{n \times n} ({\bf v})$ is unreachable from a shift-symmetric ${\bf s} \in S_{n \times n} ({\bf v})$
\begin{equation}
\Phi^i(\mathbf{s}) \in S_{n \times n}({\bf v}) \niton {\bf q}.
\end{equation}


As a consequence, several tasks for CA are principally insolvable. For instance, a target configuration for leader election \cite{Smith71} contains exactly one cell in the leader state $a \in \Sigma$. This configuration is asymmetric for $n > 1$ as shown in Figure \ref{fig:symmetric-vs-asymmetric-examples} (asym-d), and therefore unreachable from any shift-symmetric configuration. Further, several image-processing tasks illustrated in Figure \ref{fig:symmetric-vs-asymmetric-examples} are insolvable: e.g., image translation (sym-c to asym-c) \cite{ioannidis2014}, and pattern recognition or noise filtering (sym-c to asym-f and sym-a to sym-b) \cite{Rosin2006}. Also, the task of random \cite{tomassini2000} or prime $p$ number generation is insolvable if $p\not| n$ (for $p$ = 7: sym-a to asym-e). If a configuration is  shift-symmetric the associative memory \cite{chowdhury1995} has a corrupted (non-uniform) hashing function to handle collisions (e.g., sym-a to asym-a). This in general sense also applies to encryption \cite{wang2013}.


\begin{figure}[t!]
    \begin{center}
        \includegraphics[width=0.49\textwidth]{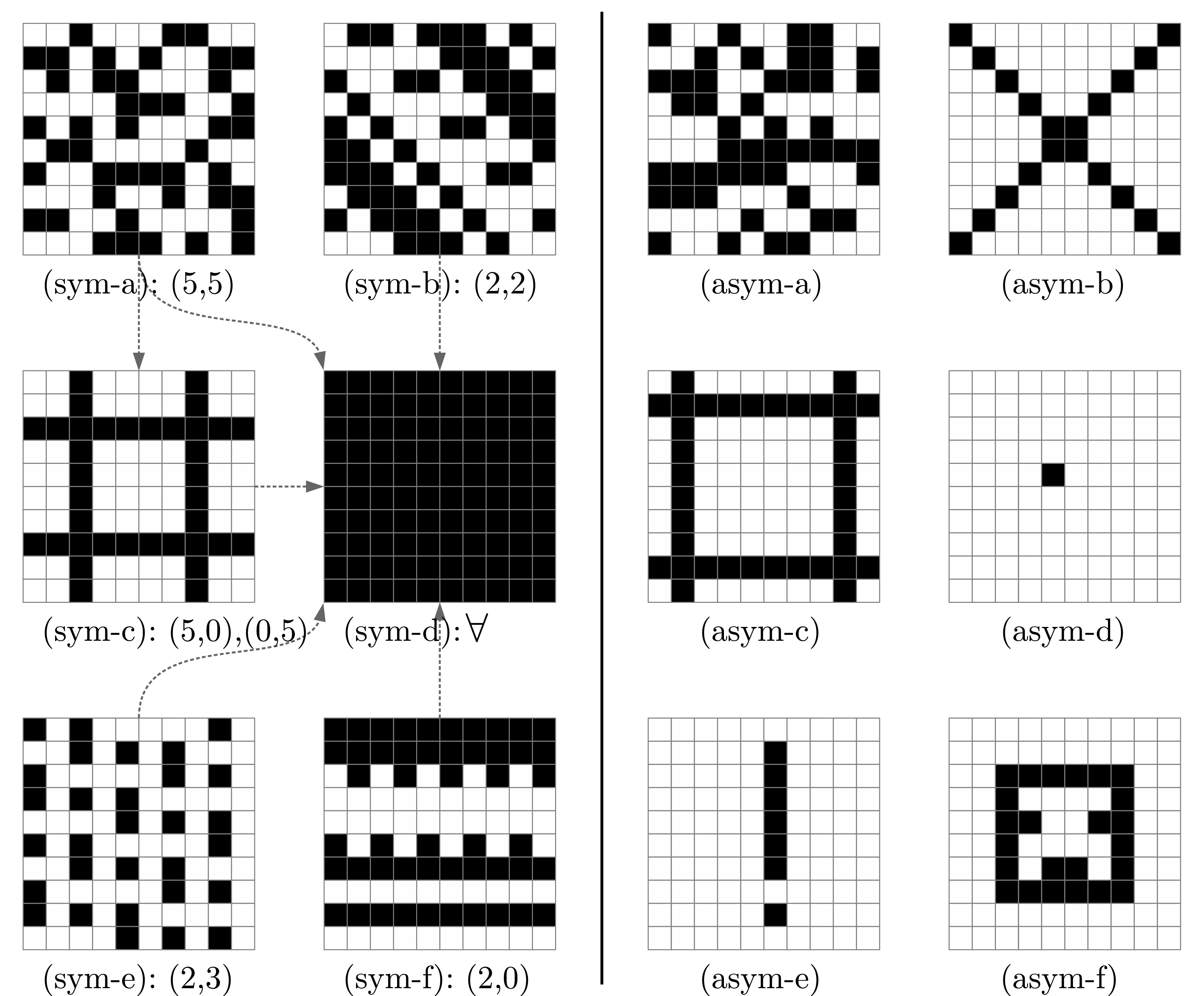}
     \end{center}
    \caption{Examples of shift-symmetric configurations with generating vector(s) in the left, and shift-asymmetric configurations in the right column, all on the lattice $Z_{10 \times 10}$. Since a shift-symmetry by a specific vector cannot be lost, there exists no CA (transition function), which would eventually reach any of the asymmetric configurations (right) from any of the symmetric ones (left). That also means a shift-symmetric configuration by a vector ${\bf u}$ cannot be reached from another shift-symmetric configuration, which is not generated by $\langle {\bf u} \rangle$. For instance the configuration sym-a, generated by the vector ${\bf v}=(5,5)$, cannot be transformed to the configuration sym-b generated by the vector ${\bf u} = (2,2)$, but can be transformed to the configuration sym-c, since its vector space $\langle (5,0),(0,5) \rangle \ni {\bf v}$. Reachability among the shift-symmetric configurations in the figure are indicated by arrows.}
    \label{fig:symmetric-vs-asymmetric-examples}
\end{figure}

%% file: 4-EnumSymmetricConfigs.tex
\section{Enumerating Shift-Symmetric Configurations}
\label{sec:enum-symmetric-configs}

In this section we will further investigate shift-symmetric two-dimensional configurations and ask how many there are in a square lattice of size $N = n^2$. First, to generalize shift-symmetry and lay a solid ground for group-centric analysis we define the symmetric configurations over several generators.

Let $\mathbb{L} \subseteq \mathbb{Z}_n \times \mathbb{Z}_n$.
We define the set of {\bf $\mathbb{L}$-symmetric configurations} to be the set
\begin{equation}
S_{n \times n}(\mathbb{L}) = \{\mathbf{s} \in \Sigma^{n \times n} \,|\, \forall {\bf u} \in \mathbb{Z}_n \times \mathbb{Z}_n,\,\forall {\bf v} \in \langle \mathbb{L} \rangle : s_{{\bf u}} = s_{{\bf u} \oplus {\bf v}}\},
\end{equation}
where $\langle \mathbb{L} \rangle = \{c_1 {\bf v}_1 \oplus \ldots \oplus c_{|\mathbb{L}|} {\bf v}_{|\mathbb{L}|} \, | \, c_i \in \mathbb{Z}_n\}$.
In other words, $S_{n \times n}(\mathbb{L})$ denotes the set of all shift-symmetric configurations of size $N = n^2$ over the alphabet $\Sigma$ with generator set $\mathbb{L}$.

 
Directly from the definition, for any subset $\mathbb{L} \subseteq \mathbb{Z}_n \times \mathbb{Z}_n$,
\begin{equation}
\label{cor:symmetric-config-set-2d-set-size}
|S_{n \times n}(\mathbb{L})| = |\Sigma |^\frac{n^2}{|\langle \mathbb{L} \rangle|},    
\end{equation}

and for any ${\bf u}, {\bf v} \in \mathbb{Z}_n \times \mathbb{Z}_n$
\begin{equation}
\label{cor:symmetric-config-set-2d-cap}
S_{n \times n}({\bf u}) \cap S_{n \times n}({\bf v}) = S_{n \times n}(\{{\bf u},{\bf v}\}).    
\end{equation}











The following equivalence, which may sound counterintuitive at first, adapts shift-symmetry for the theory of groups. It states that if a vector ${\bf v}$ generates a cyclic \textit{subgroup} of another vector ${\bf u}$, then its set of shift-symmetric configurations is a \textit{superset} (not a subset!) of that generated by the vector ${\bf u}$ and vice versa, i.e., for any ${\bf u}, {\bf v} \in \mathbb{Z}_n \times \mathbb{Z}_n$
\begin{equation-w-ref}[lemma:symmetric-config-subset-2d]
S_{n \times n}({\bf u}) \subseteq S_{n \times n}({\bf v}) \iff \langle {\bf v} \rangle \leq \langle {\bf u} \rangle.
\end{equation-w-ref}




\vspace{5pt}

Now, in a straightforward manner, we define the set of all shift-symmetric configurations $S_{n \times n}$ over all possible combinations of vectors (shifts) for a given lattice as
\begin{equation}
\label{def:symmetric-config-set-2d}
S_{n \times n} = \bigcup_{{\bf 0}\not={\bf v} \in \mathbb{Z}_n \times \mathbb{Z}_n} S_{n \times n}({\bf v}).
\end{equation}


Due to non-trivial intersections of the sets $S_{n \times n}({\bf v})$, it is fairly unpractical to count the shift-symmetric configurations over all $n^2-1$ vectors. We, instead, construct significantly fewer generators using prime factors of $n$, which equivalently produce an entire set of shift-symmetric configurations. 

We start with the definition of the generators. For any natural number $n$ let $n=\prod_{j=1}^{\omega(n)} p_j^{\alpha_j}$ be the prime factorization, where $\omega(n)$ denotes the number of distinct prime factors. We define the generator set $G_n$ as
\begin{equation}
\label{def:gen-sets}
    G_n = \bigcup_{j = 1}^{\omega(n)} G_n(p_j),
\end{equation}
where for each prime divisor $p_j$ 
\begin{equation*}
    G_n(p_j) = \left\{\left(0, \frac{n}{p_j}\right)\right\} \cup \left\{ \left(\frac{n}{p_j}, i \frac{n}{p_j}\right):0 \leq i \leq p_j-1 \right\}.
\end{equation*}

The total number of these generators is then 
\begin{equation}
|G_n| =  \omega(n) + \sum_{i = 1}^{\omega(n)} p_i.
\end{equation}




Using some divisibility arguments we can prove that they indeed produce all shift-symmetric configurations 
\begin{equation-w-ref}[lemma:symmetric-config-primes-2d]
S_{n \times n} = \bigcup_{{\bf w} \in G_n} S_{n \times n}({\bf w}).
\end{equation-w-ref}


Further, we show these prime-based generators are mutually independent, thus greatly simplifying the counting problem to relatively compact closed formulas. For any distinct ${\bf u}, {\bf v} \in G_n$ ($n \in {\mathbb{N}}$),
\begin{equation-w-ref}[lemma:symmetric-config-2d-primes-linearly-independent-all]
| \langle {\bf u} \rangle \cap \langle {\bf v} \rangle | = 1,
\end{equation-w-ref}



and for any distinct ${\bf u},{\bf v} \in G_n(p)$, where $p$ is a prime divisor of $n$ and $\hat{n}=n/p$
\begin{equation-w-ref}[lemma:symmetric-config-2d-primes-linearly-independent]
\langle {\bf u},{\bf v} \rangle = \langle (\hat{n},0), (0,\hat{n}) \rangle.
\end{equation-w-ref}
In particular, $|\langle {\bf u},{\bf v} \rangle| = p^2.$


Finally, we are ready to  enumerate shift-symmetric configurations. In the following formulas, given any ${\bf v}$, ${\bf w} \in {\mathbb{Z}}^k$, we write ${\bf v} \unlhd {\bf w}$ whenever the coordinates satisfy $v_i \leq w_i$ for every $i$ $(1 \leq i \leq k)$. We write ${\bf v} \lhd {\bf w}$ if ${\bf v} \unlhd {\bf w}$ and ${\bf v} \not= {\bf w}$. We denote the sum of the coordinates by $|{\bf v}| = \sum_{i=1}^k v_i$, and for any $m \in \mathbb{Z}$, we write ${\bf m}$ for the $k$-tuple whose coordinates all equal $m$. Let $n=\prod_{i=1}^k p_i^{\alpha_i}$ be the prime factorization of $n$, where $k=\omega(n)$, the number of distinct prime factors of $n$. Note that a one-by-one lattice offers no symmetries since there exists no non-zero shift in $\mathbb{Z}_1 \times \mathbb{Z}_1$.


As the base, we first combine the generators $G_n$ directly by the inclusion-exclusion principle, and apply the fact that these generators are mutually independent (Eq. \ref{lemma:symmetric-config-2d-primes-linearly-independent-all-back-ref}), as well as that their joint size is at most $|\langle {\bf u},{\bf v} \rangle| = p^2$ (Eq. \ref{lemma:symmetric-config-2d-primes-linearly-independent-back-ref}). That gives us
\begin{equation-w-ref}[lemma:symmetric-config-overall-2d-size]
|S_{n \times n}| = \sum_{{\bf 0} \lhd {\bf v} \unlhd {\bf p}+{\bf 1}}
(-1)^{1 + |{\bf v}|} \prod_{i=1}^{k} \binom{p_i + 1}{v_i} \left| \Sigma \right|^{f({\bf v})},
\end{equation-w-ref}
where ${\bf p} = (p_1,\ldots,p_k)$ and $f({\bf v}) = n^2 \prod_{i=1}^k p_i^{-\min (v_i,2)}.$


An alternative and more efficient counting is based on an idea of grouping of the exponential elements $\left| \Sigma \right|^{f({\bf v})}$ from the original formula (Eq. \ref{lemma:symmetric-config-overall-2d-size-back-ref}), which are costly to calculate.
\begin{equation-w-ref}[lemma:symmetric-config-overall-2d-size-alternative]
|S_{n \times n}| = \sum_{{\bf 0} \lhd {\bf v} \unlhd {\bf 2}} 
|\Sigma|^{g({\bf v})}
\left(
\sum_{{\bf v} \unlhd {\bf u} \unlhd {\rm top}({\bf v})} 
(-1)^{1 + |{\bf u}|} \prod_{i=1}^{k}\binom{p_i + 1}{u_i} \right)
\end{equation-w-ref}
where $g({\bf v}) = n^2 \prod_{i=1}^k p_i^{-v_i} \,$ and $\, {\rm top}({\bf v}) \in {\mathbb{Z}}^k$ has $i$th coordinate 
$${\rm top}(i) = \begin{cases} v_i &\mbox{if } v_i < 2 \\
p_i + 1 & \mbox{if } v_i = 2. \end{cases}$$ 


The final formula that follows is the most efficient because, besides having the exponential elements grouped, it also reduces the inner binomial sum to a simple expression $r(i)$.

\begin{equation-w-ref}[theorem:symmetric-config-overall-2d-size-final]
|S_{n \times n}| = \sum_{{\bf 0} \lhd {\bf v} \unlhd {\bf 2}}
(-1)^{1 + |{\bf v}|}
|\Sigma|^{g({\bf v})}
\prod_{i=1}^{k}r(i)
\end{equation-w-ref}
where $g({\bf v}) = n^2 \prod_{i=1}^k p_i^{-v_i}$ and
$$
r(i) =
\begin{cases} 
    1 &\mbox{if } v_i = 0 \\
    p_i + 1 & \mbox{if } v_i = 1 \\
    p_i & \mbox{if } v_i = 2.
\end{cases}
$$



Interestingly, for a prime lattice $n=p$ the vector ${\bf v}$ is $(1)$ and $g({\bf v}) = n$, or ${\bf v} = (2)$ and $g({\bf v}) = 1$, which forces the formula \ref{theorem:symmetric-config-overall-2d-size-final-back-ref} to collapse to
\begin{equation-w-ref}[cor:symmetric-config-overall-2d-prime-size]
|S_{n \times n}| = |\Sigma|^n(n + 1) - |\Sigma|n.
\end{equation-w-ref}

The gradual improvements from Eq. \ref{lemma:symmetric-config-overall-2d-size-back-ref}, \ref{lemma:symmetric-config-overall-2d-size-alternative-back-ref}, and finally to Eq. \ref{theorem:symmetric-config-overall-2d-size-final-back-ref} are illustrated for $n = 2^{\alpha_1} 3^{\alpha_2}$ in Appendix \ref{appendix:examples}.

\begin{figure}[b!]
    \begin{center}
        \begin{overpic}[width=0.48\textwidth]{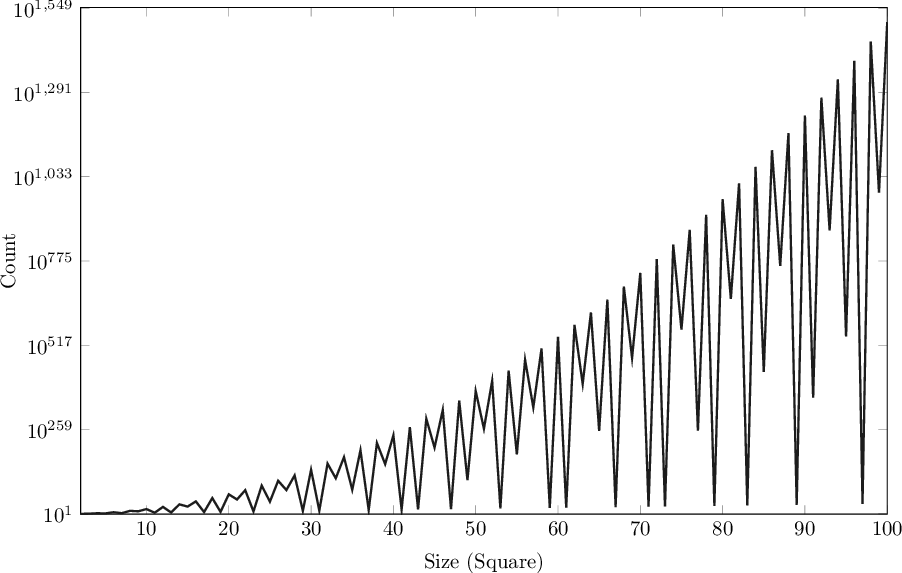}
            \put(35,70){\includegraphics[width=0.22\textwidth]{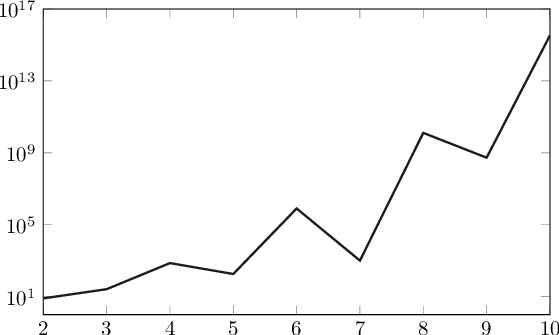}}
        \end{overpic}
    \end{center}
    \caption{Number of shift-symmetric two-dimensional binary ($|\Sigma| = 2$) configurations for the lattice sizes $2^2$ to $100^2$ with an inset focused on the area $2^2$ to $10^2$. Note the local minima for prime and local maxima for even sizes.}
    \label{fig:symmetric-count-2d}
\end{figure}

\begin{figure}[b!]
    \begin{center}
        \includegraphics[width=0.45\textwidth]{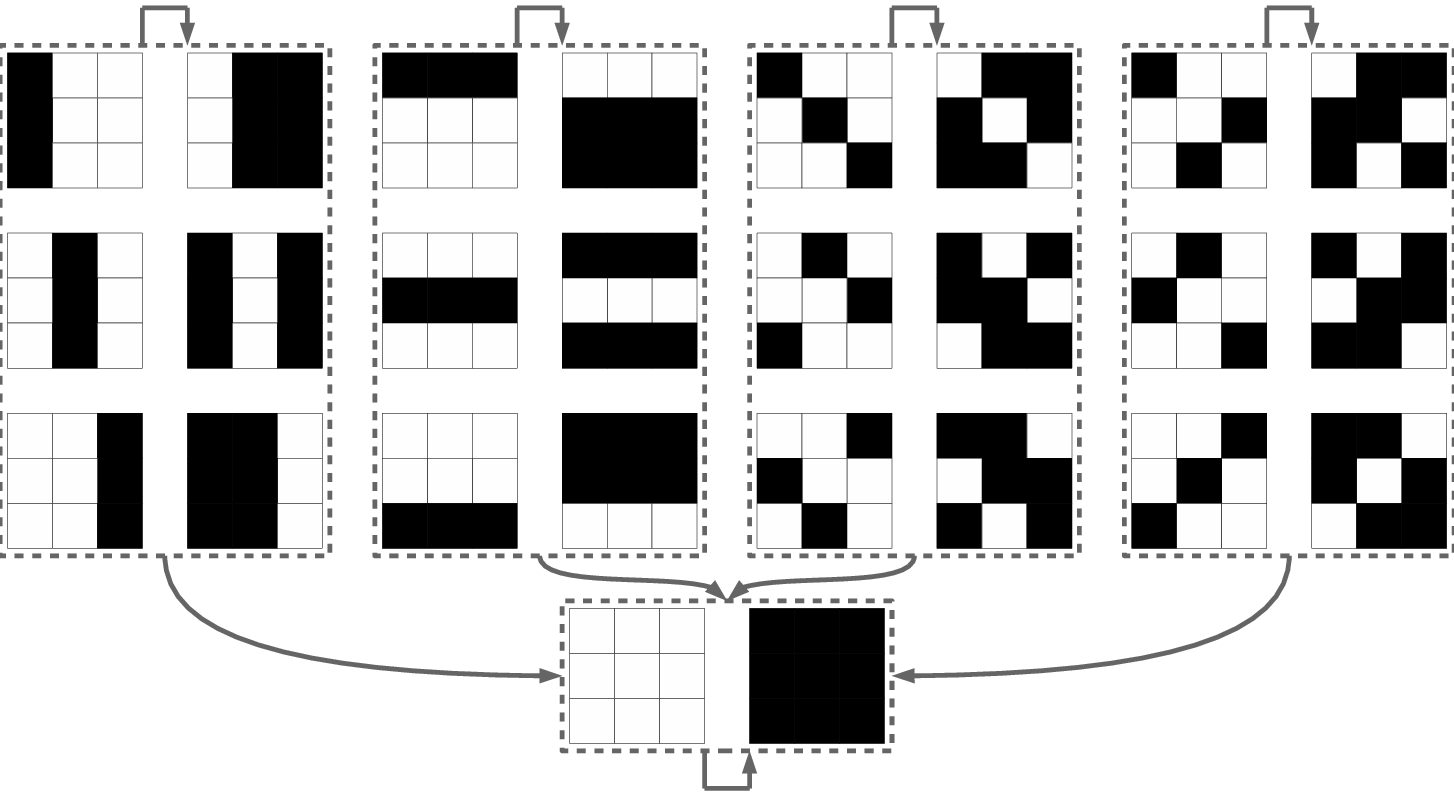}
    \end{center}
    \caption{All $26$ binary shift-symmetric configurations for the lattice size of $3^2$ grouped into 5 classes based on the generating vector(s). The vectors are from left to right: $(0,1),(1,0),(1,1),(1,2)$ and the one at the bottom containing all of them. The arrows show allowed transitions. Note that for prime-size binary lattices \mbox{$|S_{n \times n}| = 2^n(n + 1) - 2n$}.}
    \label{fig:symmetric-configs-3x3}
\end{figure}

\subsection{Bounding the Number of Shift-Symmetric Configurations}
\label{subsec:bounds-symmetric-configs}

In the previous section we derived closed and increasingly efficient formulas for counting the number of shift-symmetric configurations in a square lattice $N = n^2$. To get a deeper and more qualitative insight we now bound this number from the top and the bottom by exponential functions. We prove that the lower bound is tight on prime lattices (example in Fig. \ref{fig:symmetric-configs-3x3}), whereas local maxima are reached on even ones (Fig. \ref{fig:symmetric-count-2d}).

More precisely, for any $n \in {\mathbb{N}}$, $|S_{n \times n}|$ can be bounded as
\begin{equation-w-ref}[lemma:symmetric-config-lower-bound]
|\Sigma|^n(n + 1) - |\Sigma|n \leq |S_{n \times n}|,
\end{equation-w-ref}
where equality holds if and only if $n$ is a prime.






For an upper bound, let $n=\prod_{i=1}^k p_i^{\alpha_i}$ be the prime factorization of $n$, where $k=\omega(n)$, the number of distinct prime factors of $n$. Then
\begin{equation-w-ref}[lemma:symmetric-config-upper-bound]
|S_{n\times n}| \leq 6 \log_2(n)|\Sigma|^\frac{n^2}{2}.
\end{equation-w-ref}




Note that our bound is significantly lower than the bound $|\Sigma|^{n^2}-(n^2-1)|\Sigma|^{n^2/2}$ found by Castillo-Ramirez and Gadouleau \cite{castillo2016}. Also recall that they addressed a more general problem of counting aperiodic configurations on an arbitrary group.

By combining the inequalities \ref{lemma:symmetric-config-lower-bound-back-ref} and \ref{lemma:symmetric-config-upper-bound-back-ref} the number of shift-symmetric configurations satisfies
\begin{equation}
\label{cor:symmetric-config-lower-upper-bounds}
|\Sigma|^n(n + 1) - |\Sigma|n \leq |S_{n\times n}| \leq 6 \log_2(n)|\Sigma|^\frac{n^2}{2}.
\end{equation}

\subsection{Probability of Selecting Shift-Symmetric Configuration over Uniform Distribution}
\label{subsec:prob-symmetric-configs-uniform}

To calculate the probability that a randomly drawn configuration is shift-symmetric, we first handle a uniform distribution, in which each symbol from $\Sigma$ for $s_i$ in configuration $\mathbf{s}$ is equally likely. For non-symmetric tasks, this probability directly equals a least expected insolvability (or error lower bound).

Overall, there exist $|\Sigma|^{n^2}$ configurations and each configuration is equally likely, hence the probability of selecting a shift-symmetric configuration in a square lattice of size $N = n^2$ over uniform distribution is $P_{n \times n}^{\rm{unif}} = \frac{|S_{n \times n}|}{|\Sigma|^{n^2}}$. Further, by applying the inequality \ref{cor:symmetric-config-lower-upper-bounds} and knowing that $n|\Sigma|^n \leq |\Sigma|^n(n + 1) - |\Sigma|n$ we can bound the probability as
\begin{equation}
\label{lemma:symmetric-config-uniform-prob-lower-upper-bounds}
n|\Sigma|^{-n^2 + n} \leq P_{n \times n}^{\rm{unif}} \leq 6 \log_2(n)|\Sigma|^{-\frac{n^2}{2}}.
\end{equation}

\begin{figure}[b!]
    \begin{center}
        \begin{overpic}[width=0.48\textwidth]{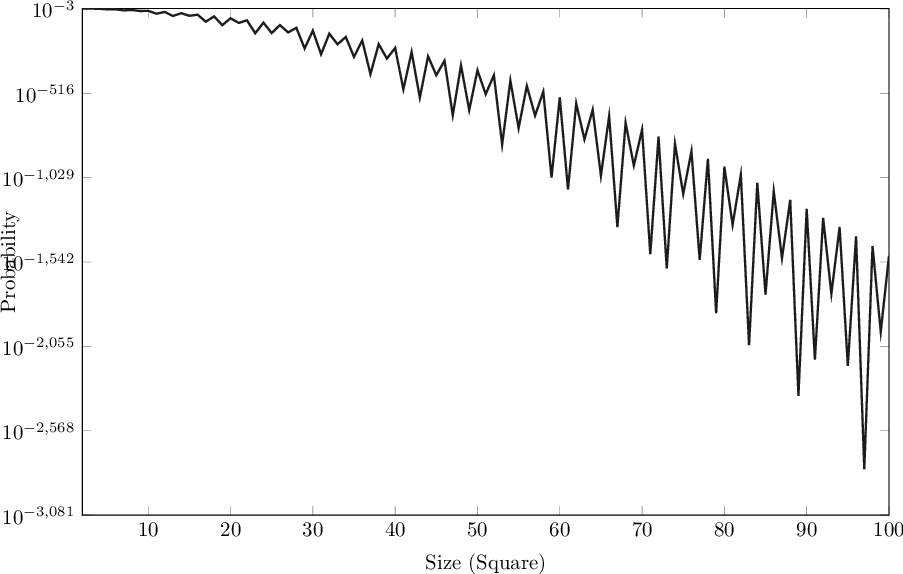}
            \put(35,30){\includegraphics[width=0.22\textwidth]{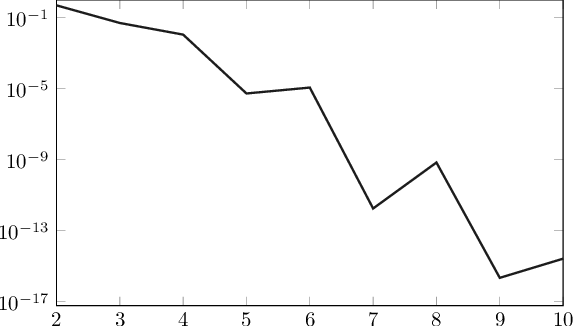}}
        \end{overpic}
    \end{center}
    \caption{Probability of selecting a shift-symmetric two-dimensional binary ($|\Sigma| = 2$) configuration using uniform distribution for the lattice sizes $2^2$ to $100^2$ with an inset focused on the area $2^2$ to $10^2$. Note the local minima for prime and local maxima for even sizes ($n > 4$).}
    \label{fig:symmetric-uniform-prob-2d}
\end{figure}

As exemplified in Figure \ref{fig:symmetric-uniform-prob-2d} and mathematically rooted in the inequality \ref{lemma:symmetric-config-uniform-prob-lower-upper-bounds}, the probability $P_{n \times n}^{\rm{unif}}$ decreases rapidly: square-exponentially by $n$ or exponentially by the lattice size $N = n^2$. Since $|S_{n \times n}|$ depends on the prime factorization of $n$ the probability is non-monotonous. Similarly to $|S_{n \times n}|$ the probability $P_{n \times n}^{\rm{unif}}$ reaches local minima for prime and local maxima for even lattices ($n > 4$).

%% file: 5-EnumSymmetricConfigsForkActiveCells.tex
\section{Enumerating Shift-Symmetric Configurations for $k$ Active Cells}
\label{sec:enum-symmetric-configs-active-cells}

Having enumerated \textit{all} shift-symmetric configurations we now tackle a subproblem of enumerating configurations with a specific number of cells in a given state, such as the state \textit{active}. The motivation behind this endeavour is to calculate the probability of selecting a shift-symmetric configuration generated by a density-uniform distribution.

We first define a set of shift-symmetric configurations with $k$ cells in a special state $a$. Formally, for any state $a \in \Sigma$ and $n,k\in {\mathbb{N}}$, we define $D^a_{n \times n,k}$ to be the set of all square configurations with exactly $k$ sites in state $a$:
\begin{equation}
D^a_{n\times n,k} = \{ \mathbf{s} \in \Sigma^{n \times n} \, | \, _{\#_{a}} \mathbf{s} = k\},
\end{equation}

where ${\#_{a}} \mathbf{s}$ denotes the number of cells in a configuration $\mathbf{s}$ that are in a state $a$. Accordingly, let $S^a_{n \times n,k}$ be the set of such configurations that are symmetric:
\begin{equation}
\label{def:symmetric-config-2d-active-cell}
S^a_{n \times n,k} = S_{n \times n} \cap D^a_{n \times n,k},
\end{equation}

and for any ${\bf v} \in \mathbb{Z}_n \times \mathbb{Z}_n$, let $S^a_{n \times n,k}({\bf v})$ denote the set of configurations in $S^a_{n \times n,k}$ that are generated by ${\bf v}$, so that
\begin{equation}
S^a_{n \times n,k}({\bf v}) = S_{n \times n}({\bf v}) \cap D^a_{n \times n,k}.
\end{equation}



As a direct corollary, for any $a \in \Sigma$, any $n,k \in \mathbb{N}$, and ${\bf v} = (l_1, l_2) \in \mathbb{Z}_n \times \mathbb{Z}_n$,
$$S^a_{n \times n,k}({\bf v}) \neq \emptyset$$

\begin{equation}
\label{cor:symmetric-config-set-2d-active-cell-nonempty}
\mbox{iff } |\langle {\bf v} \rangle| = \frac{n}{{\rm gcd}(l_1,l_2,n)} \mbox{ is an integer that divides $k$}.
\end{equation}

To launch our enumeration endeavour, we focus first on shift-symmetric configurations of a single generating vector. For any $a \in \Sigma$, any $k \in \mathbb{N}$, and ${\bf v} \in \mathbb{Z}_n \times \mathbb{Z}_n$ such that $|\langle {\bf v} \rangle|$ divides $k$
\begin{equation-w-ref}[lemma:symmetric-config-2d-active-cell-intersect-size]
\left| S^a_{n \times n,k}({\bf v}) \right| = \bigg( \binomempty{\frac{n^2}{|\langle {\bf v} \rangle|}}{\frac{k}{|\langle {\bf v} \rangle|}} \bigg) (|\Sigma| - 1)^{\frac{n^2 - k}{|\langle {\bf v} \rangle|}}.
\end{equation-w-ref}



To derive the counting formulas for the specifics of $k$-active-cell configurations we mimic the advancements of
the three counting techniques based on mutually-independent generators for $|S_{n\times n}|$ from Section \ref{sec:enum-symmetric-configs}, but this time, we root them into Eq. \ref{lemma:symmetric-config-2d-active-cell-intersect-size-back-ref}.

As before we start with a base formula. Pick $n,k \in {\mathbb{N}}$ with $k \leq n$ and let $d={\rm{gcd}}(k,n)$.
Let $n =\prod_{i=1}^{\omega(n)} p_i^{\alpha_i}$, $k=\prod_{i=1}^{\omega(k)} q_i^{\beta_i}$, and $d = \prod_{i=1}^{\omega(d)} r_i^{\gamma_i}$ be the prime factorizations of $n$, $k$, $d$, respectively. Then for any $a \in \Sigma$,
\begin{equation-w-ref}[lemma:symmetric-config-2d-active-cell-size]
\begin{split}
|S^a_{n \times n,k}| = \sum_{{\bf 0} \lhd {\bf u} \unlhd {\bf r} + {\bf 1}}
(-1)^{1 + |{\bf u}|} \left( \prod_{i=1}^{\omega(d)} \binom{r_i + 1}{u_i} \right) \\
\times \binom{
\frac{n^2}{h({\bf u})}}{
\frac{k}{h({\bf u})}} (\left| \Sigma \right|-1)^{\frac{n^2-k}{h({\bf u})}},
\end{split}
\end{equation-w-ref}
where ${\bf r} = (r_1,\ldots,r_{\omega(d)}) \,$ and $\, h({\bf u}) = \prod_{i=1}^{\omega(d)} r_i^{\min (u_i,2)}.$



Similarly to Eq. \ref{lemma:symmetric-config-overall-2d-size-alternative-back-ref} from Section \ref{sec:enum-symmetric-configs}, the following alternative counting method is more efficient than the core Eq. \ref{lemma:symmetric-config-2d-active-cell-size-back-ref} due to the grouping of the exponential elements. 
\begin{equation-w-ref}[lemma:symmetric-config-2d-active-cell-size-alternative]
\begin{split}
|S^a_{n \times n, k}| = \sum_{\substack{
{\bf 0} \lhd {\bf v} \unlhd {\bf 2} \\
{\bf v} \unlhd {\bf u} \unlhd {\rm top}({\bf v})}}
(-1)^{1 + |{\bf u}|}
\binom{\frac{n^2}{h({\bf v})}}{\frac{k}{h({\bf v})}}
(|\Sigma|-1)^{\frac{n^2-k}{h({\bf v})}} \\
\times \prod_{i=1}^{\omega(d)}\binom{r_i + 1}{u_i}, 
\end{split}
\end{equation-w-ref}
where $h({\bf v}) = \prod_{i=1}^{\omega(d)} r_i^{\min(v_i,2)} \,$ and $\, {\rm top}({\bf v}) \in {\mathbb{Z}}^{\omega(d)}$ has $i$th coordinate $${\rm top}(i) = \begin{cases} v_i &\mbox{if } v_i < 2 \\
r_i + 1 & \mbox{if } v_i = 2. \end{cases}$$


At last, as a parallel to Eq. ~\ref{theorem:symmetric-config-overall-2d-size-final-back-ref}
we derive the final formula, which further simplifies the counting mechanics by collapsing the inner binomial sum to a simple expression $r(i)$. 
\begin{equation-w-ref}[theorem:symmetric-config-2d-active-cell-size-final]
|S^a_{n \times n, k}| = \sum_{{\bf 0} \lhd {\bf v} \unlhd {\bf 2}}
(-1)^{1 + |{\bf v}|}
\binom{\frac{n^2}{h({\bf v})}}{\frac{k}{h({\bf v})}}
(|\Sigma|-1)^{\frac{n^2-k}{h({\bf v})}} \prod_{i=1}^{\omega(d)}r(i), 
\end{equation-w-ref}
where $h({\bf v}) = \prod_{i=1}^{\omega(d)} r_i^{\min(v_i,2)} \,$ and $$ r(i) =
\begin{cases} 
    1 &\mbox{if } v_i = 0 \\
    p_i + 1 & \mbox{if } v_i = 1 \\
    p_i & \mbox{if } v_i = 2.
\end{cases}$$


As a special case, it can be shown that the number of \textit{binary} symmetric configurations ($|\Sigma| = 2$) with $k$ sites in state $a$ is
\begin{equation}
\label{cor:symmetric-config-2d-active-cell-size-2-back-ref}
|S^a_{n \times n, k}| = \sum_{{\bf 0} \lhd {\bf v} \unlhd {\bf 2}}
(-1)^{1 + |{\bf v}|}
\binom{\frac{n^2}{h({\bf v})}}{\frac{k}{h({\bf v})}} \prod_{i=1}^{\omega(d)}r(i).
\end{equation}


For illustration purposes, an example of the three increasingly more compact counting formulas is given for $n = 2^{\alpha_1} 3^{\alpha_2}$ and $k = 2^{\beta_1} 3^{\beta_2}, \beta_1 \leq \alpha_1, \beta_2 \leq \alpha_2$ in Appendix \ref{appendix:examples}. 

\begin{figure}[t!]
    \begin{center}
        \begin{overpic}[width=0.48\textwidth]{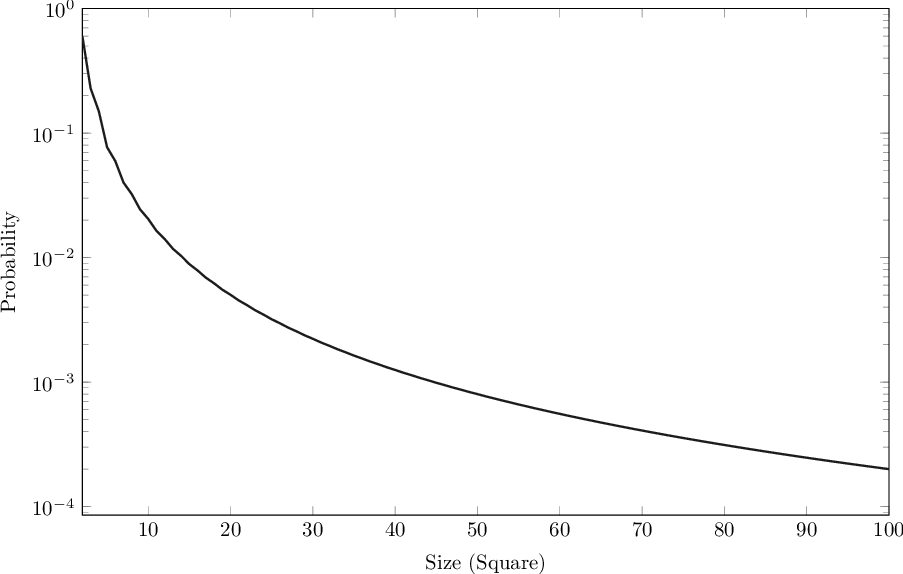}
            \put(115,70){\includegraphics[width=0.22\textwidth]{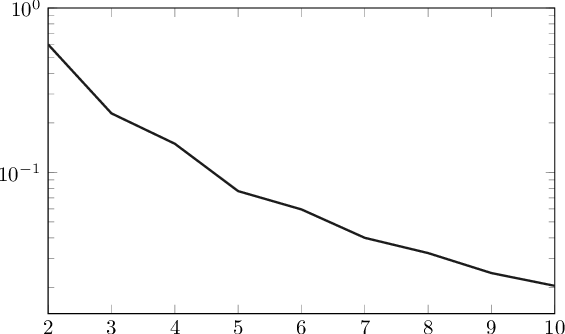}}
        \end{overpic}
    \end{center}
    \caption{Probability of selecting a shift-symmetric two-dimensional binary ($|\Sigma| = 2$) configuration using density-uniform distribution for the lattice sizes $2^2$ to $100^2$ with an inset focused on $2^2$ to $10^2$.}
    \label{fig:symmetric-density-uniform-prob-2d}
\end{figure}

\subsection{Probability of Selecting Shift-Symmetric Configuration over Density-Uniform Distribution}
\label{subsec:prob-symmetric-configs-density-uniform}

Besides a uniform distribution, a CA's performance is commonly evaluated using a so-called density-uniform distribution, in which the probability of selecting $k$ active cells ($_{\#_{a}} \mathbf{s} = k$), a \textit{density}, is uniformly distributed. Since for a density $k$ there exist $\binom{n^2}{k}(|\Sigma| - 1)^{n^2-k}$ configurations and each density is equally likely, the probability of selecting a shift-symmetric configuration in a lattice $N = n^2$ over a density-uniform distribution is then
\begin{equation}
    P_{n \times n}^{\rm{dens}} = \frac{1}{n^2+1} \sum_{k = 0}^{n^2} \frac{|S^a_{n \times n, k}|}{\binom{n^2}{k}(|\Sigma| - 1)^{n^2-k}}. 
\end{equation}


As presented in Figure ~\ref{fig:symmetric-density-uniform-prob-2d}, the probability for density-uniform distribution decreases a magnitude slower than for the uniform one and reaches $0.001$ even for $N = 45^2$. That is due to the fact that density-uniform distribution selects configurations with a few or many active cells, which are combinatorially more symmetric, more often.

%% file: 6-DetectSymmetricConfig.tex
\section{Shift-Symmetric Configuration Detection}
\label{sec:detect-symmetric-config}

For practical reasons, e.g., to test whether a current system's configuration is shift-symmetric, and if yes take an action (restart), we provide an algorithm to effectively detect an occurrence of shift-symmetry.

First, to find out whether a configuration is shift-symmetric by a shift $\mathbf{v}$ we start at a corner cell $\mathbf{w} = (0,0)$ and check if all the cells at the orbit $\mathbf{w} \oplus i \mathbf{v}$ are in the same state. If yes, we repeat this process for the next orbit and so on, moving in an arbitrary but fixed order (e.g., left-right up-down), until we check all the cells. If a cell has been visited before we skip it and move on until we find an unvisited cell, which marks a start of the next orbit. Also, if the test fails at any point, a configuration is non-shift-symmetric (by $\mathbf{v}$), and the process can be terminated. Otherwise, the property holds for all the cells and a configuration is shift-symmetric.

To determine whether a configuration is shift-symmetric globally, a naive way would be to try all possible non-zero vectors $\mathbf{v}$ and check if any of them passes the aforementioned procedure. Luckily, as we discovered in Section \ref{sec:symmetric-configs} each configuration shift-symmetry ``overlaps'' with mutually-independent generators from $G_n$. Recall that these generators are defined by prime factors and their total number $|G_n| = \omega(n) + \sum_{i = 1}^{\omega(n)} p_i$ is significantly smaller than $n^2$.

In the worst case, the shift-symmetry test needs to visit all $n^2$ cells and there are $|G_n|$ vectors to try. Since $O(|G_n|) = O({\rm sopf}(n))$ and ${\rm sopf}(n) = \sum_{i = 1}^{\omega(n)} p_i$, also known as the integer logarithm, is at most $n$, the worst-case time complexity is
\begin{equation-w-ref}[theorem:detection-algorithm-worst-case-complexity]
O(n^3).
\end{equation-w-ref}



Similarly, with a slightly more complicated proof, we can show that the average-case time complexity of the shift-symmetry test for a configuration generated from a uniform distribution is
\begin{equation-w-ref}[theorem:detection-algorithm-average-case-complexity]
O(n^2).
\end{equation-w-ref}

Note that the worst and average-case time complexity of $O(n^3)$ and $O(n^2)$ respectively translate to $O(\sqrt{N}N)$ and linear $O(N)$ when interpreted by the optics of the number of cells $N = n^2$. The function ${\rm sopf}(n)$, which plays a crucial role in both O formulas, is of a logarithmic nature in ``most of the cases,'' but $n$ for primes. Since the number of primes is infinite we could not use any tighter asymptote than $n$. However, for randomly chosen $n$ the expected time complexities drop to just $O(log(n)n^2)$ and $O(log^2(n))$ respectively.


It is worth mentioning that the presented algorithm detects \textit{if} a configuration is shift-symmetric but does not count the number of shift-symmetries in a configuration. The validity of the detection holds because we know that \textit{any} shift-symmetric configuration must obey at least one of the prime generators from $G_n$. Nevertheless, to determine the number of shift-symmetries, i.e., the number of vectors with distinct vector spaces in $\mathbb{Z}_n \times \mathbb{Z}_n$ for which the cells at a same orbit share the same state, we would need to consider also sub-vectors, whose satisfiability cannot be generally inferred from the prime generators. Construction of a \textit{counting} algorithm is addressable but goes beyond the scope of this paper.

%% file: 7-Conclusion.tex
\section{Discussion and Conclusion}
\label{sec:conclusion}

Shift-symmetry, as we illustrated in the paper, decreases the system's computational capabilities and expressivity, and is generally good to be avoided. For each shift-symmetry, a system falls into, a configuration folds by the order of symmetry and ``independent" computation shrinks to a smaller, prime fraction of the system. The rest is mirrored and lacks any intrinsic computational value or novelty. The number of reachable configurations shrinks proportionally as well.

One of the key aspects of shift-symmetry is that it is maintained (irreversible) for any number of states, and any uniform transition and neighborhood functions. It means that the occurrence of shift-symmetry is rooted in the CA model itself, specifically, in the cells' uniformity, synchronous update, and toroidal topology. Shift-symmetry is preserved as along as a transition function is uniform (shared among the cells), even if non-deterministic. In other words, during each step a transition function can be discarded and regenerated at random. However, within the same synchronous update it must be consistent, i.e., two cells whose neighborhood's sub-configs are the same must be transitioned to the same state.

We showed that a non-shift-symmetric solution is unreachable from a shift-symmetric configuration. Even more, a shift-symmetric configuration cannot be reached from another shift-symmetric one, if the vector space defining the symmetries of the starting configuration is not a subset of the target configuration's vector space. This renders the tasks, such as leader election \cite{Smith71,Banda11}, several image processing routines including pattern recognition \cite{Rosin2006}, and encryption \cite{wang2013}, insolvable by uniform CAs in a general sense.  These procedures are fundamental for many distributed protocols and algorithms. Additionally, leader election contributes to decision making of biological societies \cite{Con08,Lus09}, and is a key driver of cell differentiation \cite{Law92,Nag03} responsible for their structural heterogeneity and specialization.

To determine how likely a configuration randomly generated from a uniform distribution is shift-symmetric, hence insolvable, we efficiently enumerated and bounded the number of shift-symmetric configurations using mutually independent generators. We also introduced a lower, tight prime-size bound, and an upper bound, and showed that even-size lattices are locally most likely shift-symmetric. By specializing on Cartesian powers of cyclic groups (two-dimensional case), we obtained more effective counting and probability formulas and sharper bounds compared to the state-of-the-art work addressing the problem for general groups \cite{castillo2016,gao2016}. We also extended our machinery to a fixed number of active symbols and derived a probability formula for density-uniform distribution.

Overall, shift-symmetry is not as rare as one would think, especially for small or non-prime lattices, or when a configuration is generated using density-uniform distribution. Asymptotically, the probability for uniform distribution drops exponentially with the lattice size but a magnitude slower for a density-uniform distribution. For instance the probability for a $100^2$ square lattice is around $10^{-1505}$ using uniform and $2 \times 10^{-4}$ using density-uniform distribution. Importantly, shift-symmetry does not necessarily have to be harmful for all the tasks. For instance, the density classification \cite{MCD96,CMD03,cenek09,reynaga2003}, which is widely used as a CA benchmark problem, requires a final configuration to be either $1^N$ if the majority of cells are initially in the state $1$, and $0^N$ otherwise. Since the expected homogeneous configurations are fully shift-symmetric, they can be reached potentially from any configuration. Naturally, that depends on the structure of a transition function but shift-symmetry does not impose any strong restrictions here. The ability of reaching a \textit{valid} answer does not necessarily mean reaching a \textit{correct} answer. However, for the density classification, shift-symmetry tolerates the latter as well. It is because a shift-symmetric configuration consists purely of repeated sub-configurations, and so the density (ratio of ones) in a sub-configuration is the same as in the whole.

To detect whether a configuration is shift-symmetric we constructed an algorithm, which, by using the base prime generators, can effectively determine a presence of shift-symmetry in linear $O(N)$ time for prime and just $O((\frac{1}{2}\log(N))^2)$ for randomly chosen $N$ on average.


By moving from one to two dimensions we generalized our machinery to vector translations, which can be extended to the $n$-dimensional case \cite{Brunnet2004}. It is expected that the number of shift-symmetric configurations will grow with the dimensionality of lattice. It will be interesting to investigate this relation from the perspective of prime-exponent divisors.


An important implication of shift-symmetry is that cyclic behavior must occur only within the same symmetry class defined by a set of prime shifts (vectors) as illustrated in Figure \ref{fig:symmetric-configs-3x3}. Note that we count no-symmetry as a class as well. This leads to the realization that once a CA gains a symmetry, i.e., a configuration crosses symmetry classes, it cannot be injective and reversible, and there must exist a configuration without a predecessor, a so-called ``Garden of Eden" configuration \cite{amoroso1970,kari1990}. It means that the only way for the CA to stay injective is to decompose all the configurations into cycles, each fully residing in a certain shift-symmetry class. Again one large class would contain all the non-shift-symmetric configurations. Open question is for which lattices, i.e., for how many shift-symmetric configurations, CAs are non-injective, thus irreversible, on average. As opposed to our shift-symmetric endeavour, which applies to any transition function, investigating injectivity would require to assume something about the transition function, e.g., that is generated randomly. Trivially, for any lattice there always exists an injective transition function. An example is an identity function. 


As we proved, the number of symmetries in any synchronous toroidal CA is non-decreasing. A natural question is: could it be increasing in the ``average" case for a random transition function? We know that the expected behavior of randomly generated CA is most likely chaotic and the attractor length is exponential to the lattice size $N$, as opposed to ordered or complex CAs with linear or quadratic attractors \cite{wuensche1999}. Would the length of attractor be sufficient to discover a shift-symmetry if we keep a random CA running long enough, potentially $|\Sigma|^N$ time steps? As seen in Figure \ref{fig:symmetric-uniform-prob-2d}, the ratio of shift-symmetric configurations assuming a uniform distribution is exponentially decreasing with the lattice size, and prime lattices could produce ``only" around $n|\Sigma|^n$ symmetric configurations. For a randomly chosen lattice size, dimensions, and cell connectivity, we expect the number of reachable symmetries to be significantly smaller than the total number of symmetries available. However, for symmetry-rich lattices, we speculate that toroidal synchronous uniform systems, such as CAs, could undergo \textit{spontaneous symmetrization} contracting an initial configuration to a fully homogeneous state (analogical to Big Crunch). If proven, it would directly imply the system's non-injectivity and irreversibility, and would bind symmetrization with \textit{non-ergodicity}. This hypothesis will be addressed in our future work.






We suggest that several phenomena observed in CA dynamics, such as irreversibility, emergence of structured ``patterns", and self-organization could be explained or contributed to shift-symmetry. As demonstrated by Wolfram \cite{wolfram1983} on 256 elementary one-dimensional CAs, when run long enough, most of these CAs condensate to ordered structures: homogeneous configurations and self-similar patterns, which are in fact shift-symmetric. 


A straightforward way to fight symmetry would be to introduce noise, i.e., to break the uniformity of cells and/or to use an asynchronous update. Based on the amount of noise, this could, however, disrupt the consistency of local, particle-based, interactions, which give rise to a global computation.
Clearly, asynchronicity makes a system more robust but sacrifices 
the information processing by \textit{algebraic} structures, which could exist only due to synchronous update.


Practical utility of the presented enumeration formulas and probability calculations for a given distributed application is that, we can minimize a likelihood of shift-symmetry-caused insolvability as well as the number of resources needed. An online supplementary web page, which implements these formulas as well as an embedded simulator to run a CA on a shift-symmetric configuration, can be found at \url{https://coel-sim.org/symmetry}.







%% file: 8-Proofs.tex
\section{Proofs}
\label{appendix:proofs}

Note: lemmas, theorems, and corollaries are numbered such that those starting with \textbf{E} correspond to the equations in the main text which they are referenced from (e.g., Lemma E.4 $\leftrightarrow$ Equation 4). The remaining non-equation-referenced (auxiliary) lemmas have the \textbf{A} prefix.


\begin{lemma}
\label{lemma:symmetric-config-set-2d-v-group}
For any non-zero vector (generator) ${\bf v} \in \mathbb{Z}_n \times \mathbb{Z}_n$,
$$
S_{n \times n}({\bf v}) = \{\mathbf{s} \in \Sigma^{n \times n} \,|\, \forall {\bf u} \in \mathbb{Z}_n \times \mathbb{Z}_n \forall {\bf w} \in \langle {\bf v} \rangle : s_{{\bf u}} = s_{{\bf u} \oplus {\bf w}}\},$$
where $\langle {\bf v} \rangle$ is the cyclic subgroup of  $\mathbb{Z}_n \times \mathbb{Z}_n$ gen. by ${\bf v}$. \hfill $\Box$ 
\end{lemma}


\setcounter{theorem}{3}

\begin{lemma-w-back-ref}[lemma:symmetric-config-set-2d-l1_l2-size]
For any non-zero ${\bf v} = (l_1, l_2) \in \mathbb{Z}_n \times \mathbb{Z}_n$, the following hold:

\smallskip

(i). $\; |S_{n \times n}({\bf v})| = |\Sigma|^{\frac{n^2}{|\langle {\bf v} \rangle|}} .$

\smallskip

(ii). $\; | \langle {\bf v} \rangle | = \frac{n}{\gcd (l_1, l_2, n)}$,

\smallskip

(iii). $\; |S_{n \times n}({\bf v})| = |\Sigma|^{n \, {\rm gcd}(l_1,l_2,n)}.$
\end{lemma-w-back-ref}
\begin{proof}
(i). When ${\bf v}=(l_1,l_2)$ is repeatedly applied to any cell in the lattice, an orbit is generated, consisting of $|\langle {\bf v} \rangle|$ cells that must share a common state for any configuration in $S_{n \times n}({\bf v})$. The number of distinct orbits of cells in the lattice is simply $\frac{n^2}{|\langle {\bf v} \rangle|}$. Any configuration in $S_{n \times n}({\bf v})$ is thus uniquely determined by choosing a state from $\Sigma$ for each orbit of cells, so (i) follows.

(ii). For $l \in \mathbb{Z}_n$ it is easily shown that $|\langle l \rangle| = \frac{n}{{\rm gcd}(l,n)}$, so
$$
|\langle {\bf v} \rangle| = {\rm lcm}\left(\frac{n}{{\rm gcd}(l_1,n)},\frac{n}{{\rm gcd}(l_2,n)}\right) = \frac{n}{{\rm gcd}(l_1, l_2, n)}, $$
where ${\rm lcm}$ denotes the least common multiple. 

(iii). By (ii), the exponent in (i) becomes
\begin{equation*}
\frac{n^2}{|\langle {\bf v} \rangle |} = \frac{n^2}{\frac{n}{{\rm gcd}(l_1,l_2,n)}} = n \, {\rm gcd}(l_1,l_2,n)
\end{equation*}
as desired. \hfill $\Box$
\end{proof}


\begin{lemma}
\label{lemma:symmetric-neighborhood-2d} 
Fix any non-zero vector ${\bf v} \in \mathbb{Z}_n \times \mathbb{Z}_n$ and any shift-symmetric square configuration ${\bf s} \in S_{n \times n} ({\bf v})$.  Then for any ${\bf w} \in \mathbb{Z}_n \times \mathbb{Z}_n$, the neighborhoods satisfy
$$\eta_{\bf w}({\bf s}) = \eta_{{\bf w} \oplus {\bf v}}({\bf s}).$$
\end{lemma}
\begin{proof}
Suppose the neighborhood function, which is uniformly shared by all cells, is defined by (relative) vectors ${\bf u}_1, \ldots, {\bf u}_m$, i.e., $\eta_{\bf w}({\bf s}) = (s_{{\bf w} \oplus {\bf u}_1}, \ldots, s_{{\bf w} \oplus {\bf u}_m})$ and assume the lemma does not hold, i.e., there exists ${\bf w}$ for which $\eta_{\bf w}({\bf s}) \neq \eta_{{\bf w} \oplus {\bf v}}({\bf s})$. Then
\begin{align*}
(s_{{\bf w} \oplus {\bf u}_1},\dots,s_{{\bf w} \oplus {\bf u}_m}) \neq (s_{({\bf w} \oplus {\bf v}) \oplus {\bf u}_1},\dots,s_{({\bf w} \oplus {\bf v}) \oplus {\bf u}_m})
\end{align*}
and so there exists some ${\bf u}_j$ such that $s_{{\bf w} \oplus {\bf u}_j} \neq s_{({\bf w} \oplus {\bf v}) \oplus {\bf u}_j}$, i.e., $s_{{\bf w} \oplus {\bf u}_j} \neq s_{({\bf w} \oplus {\bf u}_j) \oplus {\bf v}}$, which contradicts the assumption that ${\bf s} \in S_{n \times n} ({\bf v})$. \hfill $\Box$
\end{proof}


\setcounter{theorem}{4}

\begin{theorem-w-back-ref}[theorem:symmetric-stays-symmetric-2d]
If ${\bf s} \in S_{n \times n} ({\bf v})$ then $\Phi(\mathbf{s}) \in S_{n \times n} ({\bf v})$ for any uniform global transition rule $\Phi$.
\end{theorem-w-back-ref}
\begin{proof}
Suppose $\mathbf{q} = \Phi(\mathbf{s})$ is not symmetric by ${\bf v}$. Then, there exists ${\bf u} \in \mathbb{Z}^n \times \mathbb{Z}^n$, such that $q_{\bf u} \neq q_{{\bf u} \oplus {\bf v}}$. By Lemma~\ref{lemma:symmetric-neighborhood-2d}, $\eta_{\bf u}({\bf s}) = \eta_{{\bf u} \oplus {\bf v}}({\bf s})$, and so 
$$q_{\bf u} = \phi(\eta_{\bf u}({\bf s})) = \phi(\eta_{{\bf u} \oplus {\bf v}}({\bf s})) = q_{{\bf u} \oplus {\bf v}},$$ which is a contradiction. \hfill $\Box$
\end{proof}

\setcounter{theorem}{9}

\begin{lemma-w-back-ref}[lemma:symmetric-config-subset-2d]
For any ${\bf u}, {\bf v} \in \mathbb{Z}_n \times \mathbb{Z}_n$
$$S_{n \times n}({\bf u}) \subseteq S_{n \times n}({\bf v}) \iff \langle {\bf v} \rangle \leq \langle {\bf u} \rangle .$$
\end{lemma-w-back-ref}
\begin{proof}
($\Rightarrow$). Suppose 
 $S_{n \times n}({\bf u}) \subseteq S_{n \times n}({\bf v})$. Then 
 $$S_{n \times n}({\bf u}) = S_{n \times n}({\bf u}) \cap S_{n \times n}({\bf v}) = S_{n \times n}(\{{\bf u}, {\bf v}\})$$
 by Corollary (Eq.)~\ref{cor:symmetric-config-set-2d-cap}. But then $|\langle {\bf u} \rangle| = |\langle {\bf u}, {\bf v} \rangle|$ by Corollary (Eq.)~\ref{cor:symmetric-config-set-2d-set-size}, which forces $\langle {\bf u} \rangle = \langle {\bf u}, {\bf v} \rangle$, so that ${\bf v} \in \langle {\bf u} \rangle$ and $\langle {\bf v} \rangle \leq \langle {\bf u} \rangle$ as desired.

($\Leftarrow$). By way of contradiction, suppose that
$S_{n \times n}({\bf u}) \not \subseteq S_{n \times n}({\bf v})$ and 
$\langle {\bf v} \rangle \leq \langle {\bf u} \rangle$. Let $\mathbf{s} \in S_{n \times n}({\bf u})$ such that $\mathbf{s} \not \in S_{n \times n}({\bf v})$. Then $\mathbf{s}$ is symmetric under ${\bf u}$ but not under ${\bf v}$. Consequently, there exists ${\bf w} \in \mathbb{Z}_n \times \mathbb{Z}_n$ such that $s_{{\bf w}} \neq s_{{\bf w} \oplus {\bf v}}$. But $\mathbf{s} \in S_{n \times n}({\bf u})$ and ${\bf v} \in \langle {\bf u} \rangle$ by assumption, so Lemma~\ref{lemma:symmetric-config-set-2d-v-group} implies that $s_{{\bf w}} = s_{{\bf w} \oplus {\bf v}}$, which is a contradiction. \hfill $\Box$
\end{proof}


\begin{lemma}
\label{lemma:minimal-subgroup}
For any prime $p$ that divides $n$ and any $i \; (0 \leq i < n)$, the cyclic group $\langle(\frac{n}{p},i\frac{n}{p}) \rangle$ is simple, i.e., it has no nontrivial proper subgroups.
\end{lemma}
\begin{proof}
By Lemma~\ref{lemma:symmetric-config-set-2d-l1_l2-size}(ii), we see that $\langle(\frac{n}{p},i\frac{n}{p}) \rangle$ has order $p$, and by Lagrange's Theorem, any group with prime order is simple. \hfill $\Box$
\end{proof}

\begin{remark}
By swapping the coordinates, the proof applies also to each subgroup of the form $\langle(i\frac{n}{p}, \frac{n}{p}) \rangle$.
\end{remark}


\setcounter{theorem}{13}

\begin{lemma-w-back-ref}[lemma:symmetric-config-primes-2d]
Fix any natural number $n$ and let $n=\prod_{j=1}^{\omega(n)} p_j^{\alpha_j}$ be the prime factorization of $n$, where $\omega(n)$ denotes the number of distinct prime factors.
Then
$$S_{n \times n} = \bigcup_{{\bf w} \in G_n} S_{n \times n}({\bf w}),$$
where $G_n$ is defined as in Definition (Eq.) ~\ref{def:gen-sets}.
\end{lemma-w-back-ref}
\begin{proof}
($\subseteq$). 
Let $\mathbf{s} \in S_{n \times n}$, so that $\mathbf{s} \in S_{n \times n}({\bf v})$ for some nonzero ${\bf v} = (a,b) \in \mathbb{Z}_n \times \mathbb{Z}_n$. It suffices to show that $\langle {\bf w} \rangle \leq \langle {\bf v} \rangle$ for some ${\bf w} \in G_n$, since this fact, by Lemma \ref{lemma:symmetric-config-subset-2d}, implies  $S_{n \times n}({\bf v}) \subseteq S_{n \times n}({\bf w})$ and therefore $\mathbf{s} \in S_{n \times n}({\bf w})$.

Without loss of generality, we may assume ${\rm gcd}(a,b,n) = 1$. Otherwise, we simply divide everything by $d = \gcd(a,b,n)$ to obtain $\hat{{\bf v}} = (\hat{a},\hat{b})$, and $\hat{n}$, respectively. Once we show that $\langle \hat{{\bf w}} \rangle \leq \langle \hat{{\bf v}} \rangle$ for some $\hat{{\bf w}} \in G_{\hat{n}}$, we multiply throughout by $d$ to obtain the desired result.

{\bf Case 1.} Suppose ${\rm gcd}(a,n) = 1$. Then $ai \equiv_n b$ and $aj \equiv_n 1$ for some $i$, $j \in \mathbb{Z}$. Also, $n {\bf v} \equiv_n (0,0)$, so $| \langle {\bf v} \rangle |$ divides $n$.  Let $p$ be any prime divisor of $| \langle {\bf v} \rangle |$ and write $n = p m$ for some $m \in \mathbb{Z}$. Let ${\bf w} = (m, im)$ and note that ${\bf w} \in G_n(p)$. Also observe ${\bf v} = a (1,i)$ and ${\bf w} = m (1,i)$, so that
$$ mj {\bf v} = mja (1,i) = m ( 1,i ) = {\bf w}.$$
Therefore ${\bf w} \in \langle {\bf v} \rangle$ and thus $\langle {\bf w} \rangle \leq \langle {\bf v} \rangle$ as desired.

{\bf Case 2.} Suppose ${\rm gcd}(a,n) \neq 1$.  Let $p$ be any prime divisor of both $a$ and $n$, so that $a=pa'$ and $n = p n'$ for some $a', n' \in \mathbb{Z}$. Let ${\bf w} = (0,n')$ and note that ${\bf w} \in G_n(p)$. Observe that $an' = a'pn' = a'n \equiv_n 0$, so
$$n' {\bf v} = n' (a,b) = (an',bn') = (0,bn') = b {\bf w}.$$
Therefore $n' {\bf v} \in \langle {\bf w} \rangle$. But by Lemma~\ref{lemma:symmetric-config-set-2d-l1_l2-size}(ii), $| \langle {\bf w} \rangle | = p$, a prime.  So if $n' {\bf v}$ is nonzero, then it generates $\langle {\bf w} \rangle$. But $n' {\bf v}$ is indeed nonzero, since its second coordinate is $b n'$, and if $bn' \equiv_n 0$, then $n | bn'$. Dividing by $n'$, we see $p | b$.  But recall that $p$ divides $a$ and $n$, and we assumed at the beginning (without loss of generality) that $\gcd(a,b,n)=1$. So $p$ cannot divide $b$. This contradiction shows $n' {\bf v}$ is nonzero and so $n' {\bf v}$ generates $\langle {\bf w} \rangle$. Thus 
$$ {\bf w} \in \langle {\bf w} \rangle = \langle n' {\bf v} \rangle \leq \langle {\bf v} \rangle.$$
So $\langle {\bf w} \rangle \leq \langle {\bf v} \rangle$ as desired.

($\supseteq$). Immediate by Definition (Eq.)~\ref{def:symmetric-config-set-2d}.
\hfill $\Box$
\end{proof}


\begin{lemma-w-back-ref}[lemma:symmetric-config-2d-primes-linearly-independent-all]
Fix any $n \in {\mathbb{N}}$.  For any distinct ${\bf u}, {\bf v} \in G_n$, 
\begin{equation}
\label{eq:triv-int}
\tag{$\star$}
|\langle {\bf u} \rangle \cap \langle {\bf v} \rangle | = 1.
\end{equation}
\end{lemma-w-back-ref}
\begin{proof}
First, suppose that ${\bf u} \in G_n(p)$ and ${\bf v} \in G_n(q)$, where $p \not= q$. By Lemma~\ref{lemma:symmetric-config-set-2d-l1_l2-size}(i), $|\langle {\bf u} \rangle | = p$ and $ |\langle {\bf v} \rangle | =q$. Since $| \langle {\bf u} \rangle \cap \langle {\bf v} \rangle |$ must divide both of these primes, the line (\ref{eq:triv-int}) must hold as claimed.

Next, suppose ${\bf u},{\bf v} \in G_n(p)$ and write $n=\hat{n}p$ for some $\hat{n} \in {\mathbb{Z}}$. Suppose ${\bf u}=(\hat{n},i\hat{n})$ and ${\bf v}=(\hat{n},j\hat{n})$ for some $0\leq i<j < p$. If ${\bf x} \in \langle {\bf u} \rangle \cap \langle {\bf v} \rangle$ then $\exists k,l$ $(0 \leq k, l <p)$ such that ${\bf x} = k{\bf u} = l{\bf v}$. But then $(k\hat{n},ki\hat{n})=(l\hat{n},lj\hat{n}),$ so $k\hat{n} \equiv_n l\hat{n}$ and thus $k \equiv_p l$.  But also, $ki\hat{n} \equiv_n lj\hat{n}$, so that $ki \equiv_p lj$. Since $i \not\equiv_p j$, this forces $k \equiv_p 0$, so that ${\bf x}=0$ and (\ref{eq:triv-int}) must hold as claimed.

Finally, suppose ${\bf u},{\bf v} \in G_n(p)$ and suppose ${\bf u}=(0,\hat{n})$ and ${\bf v}=(\hat{n},i\hat{n})$ for some $0\leq i < p$. If ${\bf x} \in \langle {\bf u} \rangle \cap \langle {\bf v} \rangle$ then $\exists k,l$ $(0 \leq k, l <p)$ such that ${\bf x} = k{\bf u} = l{\bf v}$. But then $(0,k\hat{n})=(l\hat{n},li\hat{n}),$ so $0 \equiv_n l\hat{n}$ and thus $0 \equiv_p l$.  But also, $k\hat{n} \equiv_n li\hat{n}$, so that $k \equiv_p li$ and therefore $k \equiv_p 0$. Now ${\bf x}=0$ and (\ref{eq:triv-int}) must hold as claimed. \hfill $\Box$
\end{proof}


\begin{lemma-w-back-ref}[lemma:symmetric-config-2d-primes-linearly-independent]
Fix any $n \in {\mathbb{N}}$ and any prime divisor $p$ of $n$. Let $\hat{n}=n/p$.  Then for any distinct ${\bf u},{\bf v} \in G_n(p)$,
$$ \langle {\bf u},{\bf v} \rangle = \langle (\hat{n},0), (0,\hat{n}) \rangle.$$
In particular, $|\langle {\bf u},{\bf v} \rangle| = p^2.$
\end{lemma-w-back-ref}
\begin{proof}
($\subseteq$). First suppose ${\bf u}=(\hat{n},i\hat{n})$ and ${\bf v}=(\hat{n},j\hat{n})$ for some $0 \leq i<j<p$.  Then ${\bf u} = (\hat{n},0) + i(0,\hat{n})$ and ${\bf v} = (\hat{n},0) + j(0,\hat{n})$. So $\langle {\bf u},{\bf v} \rangle \subseteq \langle (\hat{n},0), (0,\hat{n}) \rangle$ as desired.  A similar argument holds when ${\bf u}=(\hat{n},i\hat{n})$ and ${\bf v}=(0,\hat{n})$.

($\supseteq$). Again suppose ${\bf u}=(\hat{n},i\hat{n})$ and ${\bf v}=(\hat{n},j\hat{n})$ for some $0 \leq i<j<p$. Then ${\bf u}-{\bf v} \in \langle (0, \hat{n}) \rangle$. But ${\bf u}-{\bf v} \not=0$ and $|\langle (0,\hat{n}) \rangle|=p$, so ${\bf u}-{\bf v}$ generates $\langle (0,\hat{n}) \rangle$.  Thus $(0,\hat{n}) \in \langle {\bf u}-{\bf v} \rangle \subseteq \langle {\bf u},{\bf v} \rangle$. Likewise, $(\hat{n},0) \in \langle j{\bf u}-i{\bf v} \rangle \subseteq \langle {\bf u},{\bf v} \rangle$, so the desired containment holds. A similar argument can be made when ${\bf u}=(\hat{n},i\hat{n})$ and ${\bf v}=(0,\hat{n})$, showing that $(\hat{n},0) \in \langle {\bf u}-i{\bf v} \rangle$, which implies the desired result.
\hfill $\Box$
\end{proof}

\begin{lemma-w-back-ref}[lemma:symmetric-config-overall-2d-size]
Let $n=\prod_{i=1}^k p_i^{\alpha_i}$ be the prime factorization of $n$, where $k=\omega(n)$, the number of distinct prime factors of $n$. Then
\begin{align*}
|S_{n \times n}| = \sum_{{\bf 0} \lhd {\bf v} \unlhd {\bf p}+{\bf 1}}
(-1)^{1 + |{\bf v}|} \prod_{i=1}^{k} \binom{p_i + 1}{v_i} \left| \Sigma \right|^{f({\bf v})},
\end{align*}
where ${\bf p} = (p_1,\ldots,p_k)$ and $f({\bf v}) = n^2 \prod_{i=1}^k p_i^{-\min (v_i,2)}.$
\end{lemma-w-back-ref}
\begin{proof}
By Lemma \ref{lemma:symmetric-config-primes-2d}, inclusion-exclusion, and Eq. \ref{cor:symmetric-config-set-2d-cap},
\begin{align*}
|S_{n \times n}| &= \left|\bigcup_{w \in G_n} S_{n \times n}(w)\right|
= \sum_{\emptyset \not= J \subseteq G_n} (-1)^{|J|+1}
 \left|S_{n \times n}(J)
\right| .
\end{align*}
Since $G_n = \bigcup_{j = 1}^{k} G_n(p_j)$, we have $k=\omega(n)$ sets from which to choose the elements of $J$, so
$$|S_{n \times n}|
= \sum_{\substack{
J_1 \subseteq G_n(p_1) \\
\ldots \\
J_k \subseteq G_n(p_k)}} (-1)^{1 + \sum_{i=1}^k|J_i|}
 \left| S_{n \times n}\left(\bigcup_{i=1}^k J_i \right)
\right|,$$
where the sum excludes the case when $J_i =\emptyset$ for all $i$. It follows from Eq.~\ref{cor:symmetric-config-set-2d-cap} that $S_{n\times n}(\bigcup J_i) = \bigcap S_{n \times n}(J_i)$ and so Eq.~\ref{cor:symmetric-config-set-2d-set-size} gives
$$\left| S_{n \times n}\left( \bigcup_{i=1}^k J_i \right) \right| = |\Sigma |^{\frac{n^2}{|\langle \bigcup_{i=1}^k J_i \rangle |}}.$$
But by Lemma~\ref{lemma:symmetric-config-2d-primes-linearly-independent-all} we know $| \langle J_i \rangle \cap \langle J_j \rangle |=1$ when $i \not= j$, so 
$$\left| \left\langle \bigcup_{i=1}^k J_i \right\rangle \right| = \prod_{i=1}^k | \langle J_i \rangle |.$$ 
Since $J_i \subseteq G_n(p_i)$, recall that  $\langle J_i \rangle = \langle (\frac{n}{p_i},0),(0,\frac{n}{p_i}) \rangle$ when $|J_i| \geq 2$ by Lemma \ref{lemma:symmetric-config-2d-primes-linearly-independent}. So $|\langle J_i \rangle| = 1$, $p_i$, and $p_i^2$ when $|J_i|=0$, $1$, and $\geq 2$, respectively. Therefore
\begin{align*}
 \prod_{i = 1}^k|\langle J_i \rangle|
= \prod_{i = 1}^k p_i^{\min(|J_i|,2)}
\end{align*}
Substituting all this into the expression for $|S_{n \times n}|$, we obtain
$$|S_{n \times n}|
= \sum_{\substack{
J_1 \subseteq G_n(p_1) \\
\ldots \\
J_k \subseteq G_n(p_k)}} (-1)^{1 + \sum_{i=1}^k|J_i|}\left(|\Sigma|^
 {\frac{n^2}{\prod_{i = 1}^k p_i^{\min(|J_i|,2)}}}\right)$$

Now, because the content of $J_i$ is irrelevant and we care only about the cardinality $|J_i|$, for each size $v_i = |J_i|$ we have $\binom{|G_n(p_i)|}{v_i} = \binom{p_i + 1}{v_i}$ ways of choosing $v_i$ elements from $G_n(p_i)$, which produces the final formula as required.
\hfill $\Box$
\end{proof}

\begin{lemma-w-back-ref}[lemma:symmetric-config-overall-2d-size-alternative]
Let $n=\prod_{i=1}^k p_i^{\alpha_i}$ be the prime factorization of $n$, where $k=\omega(n)$, the number of distinct prime factors of $n$. Then
an alternative counting of $|S_{n \times n}|$ is
\begin{align*}
|S_{n \times n}| = \sum_{{\bf 0} \lhd {\bf v} \unlhd {\bf 2}} 
|\Sigma|^{g({\bf v})}
\left(
\sum_{{\bf v} \unlhd {\bf u} \unlhd {\rm top}({\bf v})} 
(-1)^{1 + |{\bf u}|} \prod_{i=1}^{k}\binom{p_i + 1}{u_i} \right)
\end{align*}
where $g({\bf v}) = n^2 \prod_{i=1}^k p_i^{-v_i} \,$ and $\, {\rm top}({\bf v}) \in {\mathbb{Z}}^k$ has $i$th coordinate $${\rm top}(i) = \begin{cases} v_i &\mbox{if } v_i < 2 \\
p_i + 1 & \mbox{if } v_i = 2. \end{cases}$$
\end{lemma-w-back-ref}
\begin{proof}
We know that the exponent of each $p_i$ in $S_{n \times n}$ from Lemma \ref{lemma:symmetric-config-overall-2d-size} is at most $2$. Therefore for given $v_1, \ldots, v_k \in \{0,1,2\}$ we can combine all binomial expressions associated with $|\Sigma|^{\frac{n^2}{\prod_{i = 1}^k p_i^{v_i}}}$. If $v_i \leq 1$ then we have $\binom{p_i + 1}{v_i}$ selections from $G_n(p_i)$, and $\bigcup_{u_i = 2}^{p_i + 1} \binom{p_i + 1}{u_i}$ for $v_i = 2$. These two expressions could be generalized as $\bigcup_{u_i = v_i}^{{\rm top}(i)} \binom{p_i + 1}{u_i}$ using the \textit{top} function defined above. Therefore the total coefficient of $|\Sigma|^{\frac{n^2}{\prod_{i = 1}^k p_i^{v_i}}}$ is

$$\sum_{\substack{
v_1 \leq u_1 \leq {\rm top}(1) \\
\ldots\\
v_k \leq u_k \leq {\rm top}(k)}}
(-1)^{1 + \sum_{i = 1}^k u_i} \prod_{i=1}^k\binom{p_i + 1}{u_i}
$$
as required. \hfill $\Box$
\end{proof}

\begin{theorem-w-back-ref}[theorem:symmetric-config-overall-2d-size-final]
Let $n=\prod_{i=1}^k p_i^{\alpha_i}$ be the prime factorization of $n$, where $k=\omega(n)$, the number of distinct prime factors of $n$. Then
\begin{align*}
|S_{n \times n}| = \sum_{{\bf 0} \lhd {\bf v} \unlhd {\bf 2}}
(-1)^{1 + |{\bf v}|}
|\Sigma|^{g({\bf v})}
\prod_{i=1}^{k}r(i)
\end{align*}
where $g({\bf v}) = n^2 \prod_{i=1}^k p_i^{-v_i}$ and $$ r(i) =
\begin{cases} 
    1 &\mbox{if } v_i = 0 \\
    p_i + 1 & \mbox{if } v_i = 1 \\
    p_i & \mbox{if } v_i = 2.
\end{cases}$$
\end{theorem-w-back-ref}

\begin{proof}
For a given vector ${\bf v}$ with $v_1, \ldots, v_k \in \{0,1,2\}$ we define 
$$
b({\bf v}) =
\sum_{{\bf v} \unlhd {\bf u} \unlhd {\rm top}({\bf v})} 
(-1)^{1 + |{\bf u}|} \prod_{i=1}^{k}\binom{p_i + 1}{u_i},
$$
where ${\rm top}(i) = v_i \mbox{ if } v_i < 2 \mbox{ and } p_i + 1 \mbox{ if } v_i = 2$.

Using Lemma \ref{lemma:symmetric-config-overall-2d-size-alternative}, we are left to show that
$$
b({\bf v}) = (-1)^{1 + |{\bf v}|} \prod_{i=1}^{k}r(i).
$$
We prove it by induction on $k$. As the induction basis we choose $k = 1$, and so $n = p$, where $p$ is a prime. Since $b({\bf v}) = r(0)$ we need to confirm it equals $-1, p + 1,$ or $-p$ for three different cases of ${\bf v}$ defined by the function $r$.

If ${\bf v} = (0)$, ${\rm top}({\bf v}) = (0)$ and the only ${\bf u}$ is ${\bf u} = (0)$, which gives $b({\bf v}) = - \binom{p + 1}{0} = -1$.
If ${\bf v} = (1)$, ${\rm top}({\bf v}) = (1)$ and the only ${\bf u}$ is ${\bf u} = (1)$, and so $b({\bf v}) = \binom{p + 1}{1} = p + 1$.
If ${\bf v} = (2)$, ${\rm top}({\bf v}) = (p + 1)$ and ${\bf u}$ ranges from $(2)$ to $(p + 1)$. Therefore
\begin{align*}
b({\bf v}) &= \sum_{u_1 = 2}^{p + 1} (-1)^{1 + u_1} \binom{p + 1}{u_1}\\
&= \underbrace{\sum_{u_1 = 0}^{p + 1} (-1)^{1 + u_1} \binom{p + 1}{u_1}}_0 + \binom{p + 1}{0} - \binom{p + 1}{1}\\
&= -p
\end{align*}

For the induction step we prove:
$$
b({\bf v}) = (-1)^{1 + |{\bf v}|} \prod_{i=1}^{k}r(i) \Rightarrow
b({\bf w}) = (-1)^{1 + |{\bf w}|} \prod_{i=1}^{k + 1}r(i),
$$
where ${\bf w} = (v_1, \ldots, v_k, v_{k + 1})$.
Similarly to the induction basis we need to consider three cases for $v_{k + 1}$:\\

\noindent If $v_{k + 1} = 0$
\begin{align*}
b({\bf w}) &= \binom{p_{k + 1} + 1}{0} \underbrace{\sum_{{\bf v} \unlhd {\bf u} \unlhd {\rm top}({\bf v})} (-1)^{1 + |{\bf u}|} \prod_{i=1}^{k}\binom{p_i + 1}{u_i}}_{b({\bf v})} \\
&= (-1)^{1 + |{\bf v}|} \prod_{i=1}^{k}r(i) \mbox{\quad by induction step}\\
&= (-1)^{1 + |{\bf w}|} \prod_{i=1}^{k + 1}r(i)
\end{align*}
\noindent If $v_{k + 1} = 1$
\begin{align*}
b({\bf w}) &= - \binom{p_{k + 1} + 1}{1} \underbrace{\sum_{{\bf v} \unlhd {\bf u} \unlhd {\rm top}({\bf v})} (-1)^{1 + |{\bf u}|} \prod_{i=1}^{k}\binom{p_i + 1}{u_i}}_{b({\bf v})}\\
&= - (p_{k + 1} + 1) (-1)^{1 + |{\bf v}|} \prod_{i=1}^{k}r(i) \mbox{\quad by induction step} \\
&= (-1)^{1 + |{\bf w}|} \prod_{i=1}^{k + 1}r(i)
\end{align*}
\noindent If $v_{k + 1} = 2$
\begin{align*}
b({\bf w}) &= \underbrace{\sum_{u_{k + 1} = 2}^{p_{k + 1} + 1} (-1)^{u_{k + 1}} \binom{p_{k + 1} + 1}{u_{k + 1}}}_{p_{k + 1}} b({\bf v}) \\
&= p_{k + 1} (-1)^{1 + |{\bf v}|} \prod_{i=1}^{k}r(i) \mbox{\quad by induction step} \\
&= (-1)^{1 + |{\bf w}|} \prod_{i=1}^{k + 1}r(i) \mbox{\hspace{120pt} $\Box$}
\end{align*}
\end{proof}

\begin{corollary-w-back-ref}[cor:symmetric-config-overall-2d-prime-size]
Let $n=p$, where $p$ is a prime. Then
\begin{align*}
|S_{n \times n}| = |\Sigma|^n(n + 1) - |\Sigma|n. 
\end{align*}
\end{corollary-w-back-ref}
\begin{proof}
For ${\bf v} = (1)$, $g({\bf v}) = n$ and $b({\bf v}) = (n + 1)$, and for ${\bf v} = (2)$, $g({\bf v}) = 1$ and $b({\bf v}) = - n$. \hfill $\Box$
\end{proof}


\begin{lemma}
\label{lemma:symmetric-config-cummulative-prime-inequality}
Let $n=\prod_{i=1}^k p_i^{\alpha_i}$ be the prime factorization of $n$, where $k=\omega(n)$, the number of distinct prime factors of $n$, and for each $m\, (1 \leq m \leq k)$, let 

$$
q_{n \times n}^m =  \sum_{{\bf 0} \lhd {\bf v} \unlhd {\bf 2}}
(-1)^{1 + |{\bf v}|}
|\Sigma|^{g({\bf v})}
\prod_{i=1}^{m}r(i),
$$

where ${\bf v} \in {\mathbb{Z}}^m$ and $g({\bf v})$ and $r(i)$ are defined as before. Note that $|S_{n \times n}| = q_{n \times n}^k$. Then for $m < k$

$$q_{n \times n}^m \leq q_{n \times n}^{m + 1}.$$

\end{lemma}
\begin{proof}
Let ${\bf v} = (v_1, \ldots, v_m)$ and ${\bf w} = (v_1, \ldots, v_m, v_{m + 1})$. Then
\begin{align*}
q_{n \times n}^{m + 1} =& \sum_{{\bf 0} \lhd {\bf w} \unlhd {\bf 2}}
(-1)^{1 + |{\bf w}|}
|\Sigma|^{g({\bf w})}
\prod_{i=1}^{m + 1}r(i) \\
= & \sum_{v_{m + 1} = 0}^2 \biggl((-1)^{v_{m + 1}} r(m + 1)\\
& \hspace{40pt} \sum_{{\bf 0} \lhd {\bf v} \unlhd {\bf 2}}
(-1)^{1 + |{\bf v}|}
|\Sigma|^{g({\bf v})p_{m + 1}^{-v_{m+1}}}
\prod_{i=1}^{m}r(i)\biggl) \\
& + \sum_{v_{m + 1} = 1}^2 (-1)^{1 + v_{m + 1}} |\Sigma|^{n^2 p_{m + 1}^{-v_{m+1}}} r(m + 1)
\end{align*}

We split the expression into five parts:

$$ q_{n \times n}^{m + 1} = x_0 + x_1 + x_2 + y_1 + y_2$$
and define, for any $c \in \mathbb{R}$

$$
q_{n \times n}^m(c) =  \sum_{{\bf 0} \lhd {\bf v} \unlhd {\bf 2}}
(-1)^{1 + |{\bf v}|}
|\Sigma|^{g({\bf v})c}
\prod_{i=1}^{m}r(i),
$$

i.e., $q_{n \times n}^m = q_{n \times n}^m(1)$. Then

\begin{align*}
x_{0} &= q_{n \times n}^m \\
x_{1} &= - (p_{m + 1} + 1) q_{n \times n}^m(p_{m + 1}^{-1})\\
x_{2} &= p_{m + 1} q_{n \times n}^m(p_{m + 1}^{-2}) \\
y_{1} &= |\Sigma|^{n^2 p_{m + 1}^{-1}} (p_{m + 1} + 1)\\
y_{2} &= - |\Sigma|^{n^2 p_{m + 1}^{-2}} p_{m + 1}
\end{align*}

Now we show that

\begin{align*}
y_{1} + x_{1} + y_{2} \geq 0
\end{align*}

Let $A = n^2 p_{m + 1}^{-1}$. Then $(y_{1} + x_{1} + y_{2})(p_{m + 1} + 1)^{-1} $\\
\begin{align*}
& \geq
|\Sigma|^A - q_{n \times n}^m(p_{m + 1}^{-1}) -|\Sigma|^\frac{A}{p_{m+1}} \\
&\geq
|\Sigma|^A - \sum_{{\bf 0} \lhd {\bf v} \unlhd {\bf 2}}
|\Sigma|^{g({\bf v}) p_{m + 1}^{-1}}
\prod_{i=1}^{m}r(i)  -|\Sigma|^\frac{A}{p_{m+1}} \\
& =
|\Sigma|^A - \sum_{{\bf 0} \lhd {\bf v} \unlhd {\bf 2}}
|\Sigma|^{\frac{A}{\prod_{i=1}^{m}p_i^{v_i}}}
\prod_{i=1}^{m}r(i)  -|\Sigma|^\frac{A}{p_{m+1}} \\
&\geq
|\Sigma|^A - \prod_{i=1}^{m}(p_i + 1) \sum_{{\bf 0} \lhd {\bf v} \unlhd {\bf 2}}
|\Sigma|^{\frac{A}{\prod_{i=1}^{m}p_i^{v_i}}}
 -|\Sigma|^\frac{A}{p_{m+1}} \mbox{\footnotesize \hspace{10pt} $r(i) \leq p_i + 1$}\\
&\geq
|\Sigma|^A - n 2^m \sum_{{\bf 0} \lhd {\bf v} \unlhd {\bf 2}}
|\Sigma|^{\frac{A}{\prod_{i=1}^{m}p_i^{v_i}}}
 -|\Sigma|^\frac{A}{p_{m+1}} \mbox{\footnotesize \hspace{8pt} $p_i + 1 \leq 2 p_i$ and $\prod_{i=1}^{m}p_i < n$ }\\
& \geq 
|\Sigma|^A 
- n 2^m \sum_{{\bf 0} \lhd {\bf v} \unlhd {\bf 2}} |\Sigma|^{\frac{A}{p_l}}
-|\Sigma|^\frac{A}{p_{m+1}} \mbox{\footnotesize \hspace{27pt} $p_l = \min\{p_1,\ldots, p_m\}$}\\
& = 
|\Sigma|^A 
- 2^{\log_2(n) + m} \sum_{{\bf 0} \lhd {\bf v} \unlhd {\bf 2}} |\Sigma|^{\frac{A}{p_l}} 
-|\Sigma|^\frac{A}{p_{m+1}}\\
&\geq 
|\Sigma|^A 
- 2^{\log_2(n) + m} 3^m |\Sigma|^{\frac{A}{p_l}}
-|\Sigma|^\frac{A}{p_{m+1}} 
\mbox{\footnotesize \hspace{40pt} $|{\bf v}| = m$} \\
&\geq 
|\Sigma|^A 
- |\Sigma|^{\log_2(n) + m + 2m} |\Sigma|^{\frac{A}{p_l}} 
-|\Sigma|^\frac{A}{p_{m+1}} \mbox{\footnotesize \hspace{30pt} $|\Sigma| \geq 2$}\\
&\geq 
|\Sigma|^A - |\Sigma|^{3m + \log_2(n) + \frac{A}{p_l}} 
-|\Sigma|^\frac{A}{p_{m+1}}\\
&\geq 
|\Sigma|^A - |\Sigma|^{4\log_2(n) + \frac{A}{p_l}}
-|\Sigma|^\frac{A}{p_{m+1}} \mbox{\footnotesize \hspace{58pt} $m \leq k - 1 < \log_2(n)$ \hspace{40pt}}\\
&\geq 
|\Sigma|^A - |\Sigma|^{4\log_2(n) + \frac{A}{2}}
-|\Sigma|^\frac{A}{2} \mbox{\footnotesize \hspace{73pt} $p_l, p_{m + 1} \geq 2$}\\
&\geq 
|\Sigma|^A - 2|\Sigma|^{4\log_2(n) + \frac{A}{2}}\\
&\geq 
|\Sigma|^A - |\Sigma|^{4\log_2(n) + \frac{A}{2} + 1} \mbox{\footnotesize \hspace{90pt} $|\Sigma| \geq 2$}\\
&\geq 0.
\end{align*}


Since $x_2$ is non-negative we can conclude that

$$
q_{n \times n}^{m + 1} = \underbrace{x_0}_{q_{n \times n}^m} + \underbrace{x_1 + y_1 + y_2}_{\geq 0} + \underbrace{x_2}_{\geq 0} \geq q_{n \times n}^m \mbox{\hspace{30pt} $\Box$}
$$
\end{proof}


\setcounter{theorem}{20}

\begin{lemma-w-back-ref}[lemma:symmetric-config-lower-bound]
\begin{align*}
|\Sigma|^n(n + 1) - |\Sigma|n \leq |S_{n \times n}|,
\end{align*}
where equality holds if and only if $n$ is a prime.
\end{lemma-w-back-ref}
\begin{proof}
If $k = 1$, i.e., $n$ is a prime, the equality holds as shown in Corollary \ref{cor:symmetric-config-overall-2d-prime-size}. If $k > 1$ using Lemma \ref{lemma:symmetric-config-cummulative-prime-inequality} and $p_1 < n, p_1 \leq \frac{n}{2}$

\begin{align*}
|S_{n \times n}| &= q_{n \times n}^k \geq q_{n \times n}^{k - 1} \geq \ldots \geq q_{n \times n}^1 \\
& = |\Sigma|^{n^2 p_1^{-1}}(p_1 + 1) - |\Sigma|^{n^2 p_1^{-2}}p_1\\
& > |\Sigma|^{n^2 n^{-1}}(n + 1) - |\Sigma|^{n^2 n^{-2}}n \mbox{\hspace{60pt} $\Box$}
\end{align*}
\end{proof}


\begin{lemma}
\label{lemma:symmetric-config-upper-inequality}
Let $n=\prod_{i=1}^k p_i^{\alpha_i}$ be the prime factorization of $n$, where $k=\omega(n)$, the number of distinct prime factors of $n$. Then

$$
|S_{n\times n}| \leq 2 \sum_{i = 1}^k |\Sigma|^{n^2 p_i^{-1}}(p_i + 1).
$$

\end{lemma}
\begin{proof}
As in the proof of Lemma \ref{lemma:symmetric-config-cummulative-prime-inequality} we employ the function $q_{n \times n}$, which can be decomposed into five parts as defined earlier

$$ q_{n \times n}^{m + 1} = x_0 + x_1 + x_2 + y_1 + y_2$$

Now we show that

\begin{align*}
y_{1} \geq x_1 + x_2 + y_2
\end{align*}

Let $A = n^2 p_{m + 1}^{-1}$. Then $(y_1 - x_1 - x_2 - y_2)(p_{m + 1} + 1)^{-1} $\\
\begin{align*}
& \geq
|\Sigma|^A + \underbrace{q_{n \times n}^m(p_{m + 1}^{-1})}_{\geq 0} - q_{n \times n}^m(p_{m + 1}^{-2}) +\underbrace{|\Sigma|^\frac{A}{p_{m+1}}}_{\geq 0} \\
& \geq
|\Sigma|^A - q_{n \times n}^m(p_{m + 1}^{-2}) \\
&\geq
|\Sigma|^A - \sum_{{\bf 0} \lhd {\bf v} \unlhd {\bf 2}}
|\Sigma|^{g({\bf v}) p_{m + 1}^{-2}}
\prod_{i=1}^{m}r(i) \\
& =
|\Sigma|^A - \sum_{{\bf 0} \lhd {\bf v} \unlhd {\bf 2}}
|\Sigma|^{\frac{A}{\prod_{i=1}^{m}p_i^{v_i}p_{m + 1}}}
\prod_{i=1}^{m}r(i)\\
&\geq
|\Sigma|^A - \prod_{i=1}^{m}(p_i + 1) \sum_{{\bf 0} \lhd {\bf v} \unlhd {\bf 2}}
|\Sigma|^{\frac{A}{\prod_{i=1}^{m}p_i^{v_i}p_{m + 1}}} \mbox{\footnotesize \hspace{10pt} $r(i) \leq p_i + 1$}\\
&\geq
|\Sigma|^A - n 2^m \sum_{{\bf 0} \lhd {\bf v} \unlhd {\bf 2}}
|\Sigma|^{\frac{A}{\prod_{i=1}^{m}p_i^{v_i}p_{m + 1}}} \mbox{\footnotesize \hspace{8pt} $p_i + 1 \leq 2 p_i$ and $\prod_{i=1}^{m}p_i < n$ }\\
& \geq 
|\Sigma|^A 
- n 2^m \sum_{{\bf 0} \lhd {\bf v} \unlhd {\bf 2}} |\Sigma|^{\frac{A}{p_l p_{m + 1}}} \mbox{\footnotesize \hspace{27pt} $p_l = \min\{p_1,\ldots, p_m\}$}\\
& = 
|\Sigma|^A 
- 2^{\log_2(n) + m} \sum_{{\bf 0} \lhd {\bf v} \unlhd {\bf 2}} |\Sigma|^{\frac{A}{p_l p_{m + 1}}}\\
&\geq 
|\Sigma|^A - 2^{\log_2(n) + m}3^m
|\Sigma|^{\frac{A}{p_l p_{m+1}}} \mbox{\footnotesize \hspace{25pt} $|{\bf v}| = m$}\\
&\geq 
|\Sigma|^A - |\Sigma|^{\log_2(n) + m + 2m}
|\Sigma|^{\frac{A}{p_l p_{m+1}}} \mbox{\footnotesize \hspace{15pt} $|\Sigma| \geq 2$}
\end{align*}
\begin{align*}
&\geq 
|\Sigma|^A - |\Sigma|^{3m + \log_2(n)  + \frac{A}{p_l p_{m+1}}}\\
&\geq 
|\Sigma|^A - |\Sigma|^{4\log_2(n) + \frac{A}{p_l p_{m+1}}} \mbox{\footnotesize \hspace{24pt} $m \leq k - 1 < \log_2(n)$}\\
&\geq 
|\Sigma|^A - |\Sigma|^{4\log_2(n) + \frac{A}{6}}\mbox{\footnotesize \hspace{42pt} $p_l, p_{m + 1} \geq 2, p_l \neq p_{m + 1}$}\\
&\geq 0.
\end{align*}

Since $y_{1} \geq x_1 + x_2 + y_2$

$$
q_{n \times n}^{m + 1} = x_0 + x_1 + x_2 + y_1 + y_2 \leq x_0 + 2y_1.
$$

By substituting $x_0$ and $y_1$ we obtain a recursive inequality

\begin{align*}
q_{n \times n}^{m + 1} &\leq q_{n \times n}^m + 2
|\Sigma|^{n^2 p_{m + 1}^{-1}} (p_{m + 1} + 1) \\
&\leq q_{n \times n}^{m -1} + 2
|\Sigma|^{n^2 p_m^{-1}} (p_m + 1) + 2|\Sigma|^{n^2 p_{m + 1}^{-1}} (p_{m + 1} + 1)\\
&\ldots\\
&\leq 2 \sum_{i = 1}^{m + 1} |\Sigma|^{n^2 p_i^{-1}}(p_i + 1)
\end{align*}

To finalize the proof we use $|S_{n \times n}| = q_{n \times n}^k$. \hfill $\Box$
\end{proof}


\begin{lemma}
\label{lemma:symmetric-config-upper-inequality-prime}
Let $p$ be a prime divisor of $n$. Then
$$
|\Sigma|^{{n^2}p^{-1}} (p + 1) \leq |\Sigma|^{{n^2}(p - 1)^{-1}} p.
$$
\end{lemma}
\begin{proof}
Let $B = |\Sigma|^{{n^2}p^{-1}}$. Then

\begin{align*}
B^{\frac{p}{p - 1}}p - B(p + 1) & = B(B^{\frac{1}{p - 1}}p - (p + 1)) \\
& \geq B(|\Sigma|^{\frac{n}{p - 1}}p - (p + 1)) \\
& \geq B(2^{\frac{n}{p - 1}}p - (p + 1)) \mbox{\footnotesize \hspace{25pt} $|\Sigma| \geq 2$}\\
& \geq B(2p - (p + 1)) \mbox{\footnotesize \hspace{40pt} $\frac{n}{p - 1} \geq 1$}\\
& \geq 0 \mbox{\hspace{135pt} $\Box$}
\end{align*}
\end{proof}


\begin{corollary}
\label{cor:symmetric-config-upper-inequality-2}
Let $p$ be a prime of $n$. Then
$$
|\Sigma|^{{n^2}p^{-1}} (p + 1) \leq 3|\Sigma|^{\frac{n^2}{2}}.
$$

\end{corollary}
\begin{proof} By induction
\begin{align*}
\Sigma|^{{n^2}p^{-1}} (p + 1) 
& \leq |\Sigma|^{{n^2}(p - 1)^{-1}} p\\
& \leq |\Sigma|^{{n^2}(p - 2)^{-1}} (p - 1) \\
& \leq \ldots \leq |\Sigma|^{{n^2}2^{-1}} 3 \mbox{\hspace{100pt} $\Box$}
\end{align*}
\end{proof}


\setcounter{theorem}{21}

\begin{lemma-w-back-ref}[lemma:symmetric-config-upper-bound]
Let $n=\prod_{i=1}^k p_i^{\alpha_i}$ be the prime factorization of $n$, where $k=\omega(n)$, the number of distinct prime factors of $n$. Then

$$
|S_{n\times n}| \leq 6 \log_2(n)|\Sigma|^\frac{n^2}{2}.
$$

\end{lemma-w-back-ref}
\begin{proof}
By Lemma \ref{lemma:symmetric-config-upper-inequality} and Corollary \ref{cor:symmetric-config-upper-inequality-2}

$$
|S_{n\times n}| \leq 2 \sum_{i = 1}^k |\Sigma|^{n^2 p_i^{-1}}(p_i + 1) \leq 6k  |\Sigma|^\frac{n^2}{2} \leq 6 \log_2(n)|\Sigma|^\frac{n^2}{2}
$$\hfill $\Box$
\end{proof}


\setcounter{theorem}{28}

\begin{lemma-w-back-ref}[lemma:symmetric-config-2d-active-cell-intersect-size]
For any $a \in \Sigma$, any $k \in \mathbb{N}$, and ${\bf v} \in \mathbb{Z}_n \times \mathbb{Z}_n$ such that $|\langle {\bf v} \rangle|$ divides $k$
\begin{align*}
\left| S^a_{n \times n,k}({\bf v}) \right| = \binom{\frac{n^2}{|\langle {\bf v} \rangle|}}{\frac{k}{|\langle {\bf v} \rangle|}}(|\Sigma| - 1)^{\frac{n^2 - k}{|\langle {\bf v} \rangle|}}.
\end{align*}
\end{lemma-w-back-ref}
\begin{proof}
Let $\mathbf{s} \in S^a_{n \times n,k}({\bf v})$. Then the number of selections of state in $\mathbf{s}$, i.e., the pattern size, is $n^2/|\langle {\bf v} \rangle|$. To enumerate the number of such configurations, we first have to 
choose $k/|\langle {\bf v} \rangle|$ out of $n^2/|\langle {\bf v} \rangle|$ sites to be in state $a$, and then fill the remaining $n^2/|\langle {\bf v} \rangle| - k/|\langle {\bf v} \rangle|$ sites with states from $\Sigma \setminus \{a \}$. 
\hfill $\Box$
\end{proof}


\setcounter{theorem}{29}

\begin{lemma-w-back-ref}[lemma:symmetric-config-2d-active-cell-size]
Pick $n,k \in {\mathbb{N}}$ with $k \leq n$ and let $d={\rm{gcd}}(k,n)$.
Let $n =\prod_{i=1}^{\omega(n)} p_i^{\alpha_i}$, $k=\prod_{i=1}^{\omega(k)} q_i^{\beta_i}$, and $d = \prod_{i=1}^{\omega(d)} r_i^{\gamma_i}$ be the prime factorizations of $n$, $k$, $d$, respectively. Then for any $a \in \Sigma$,
\begin{align*}
|S^a_{n \times n,k}| = \sum_{{\bf 0} \lhd {\bf u} \unlhd {\bf r} + {\bf 1}}
(-1)^{1 + |{\bf u}|} \left( \prod_{i=1}^{\omega(d)} \binom{r_i + 1}{u_i} \right)\\
\times \binom{
\frac{n^2}{h({\bf u})}}{
\frac{k}{h({\bf u})}} (\left| \Sigma \right|-1)^{\frac{n^2-k}{h({\bf u})}},
\end{align*}
where ${\bf r} = (r_1,\ldots,r_{\omega(d)}) \,$ and $\, h({\bf u}) = \prod_{i=1}^{\omega(d)} r_i^{\min (u_i,2)}.$
\end{lemma-w-back-ref}

\begin{proof}
Using Eq.~\ref{def:symmetric-config-2d-active-cell}, Lemma \ref{lemma:symmetric-config-primes-2d}, and Eq. \ref{cor:symmetric-config-set-2d-active-cell-nonempty},
\begin{align*}
S^a_{n\times n,k} &= \left(\bigcup_{{\bf w} \in G_n} S_{n \times n}({\bf w}) \right) \bigcap  D^a_{n \times n,k}\\
                &= \bigcup_{i = 1}^{\omega(n)} \bigcup_{{\bf w} \in G_n(p_i)} S^a_{n \times n,k}({\bf w})\\
                &= \bigcup_{i = 1}^{\omega(d)} \bigcup_{{\bf w} \in G_n(r_i)} S^a_{n \times n,k}({\bf w}).
\end{align*}
By the inclusion-exclusion principle
\begin{align*}
|S^a_{n \times n, k}|
= \sum_{\substack{
J_1 \subseteq G_n(r_1) \\
\ldots \\
J_{\omega(d)} \subseteq G_n(r_{\omega(d)})}} (-1)^{1 + \sum_{i=1}^{\omega(d)}|J_i|}
 \left| \bigcap_{{\bf w} \in \cup_i J_i} S^a_{n \times n}({\bf w})
\right|.
\end{align*}
Now, by Eq. \ref{cor:symmetric-config-set-2d-cap}
\begin{align*}
\bigcap_{{\bf w} \in \cup_{i = 1}^{\omega(d)} J_i} S^a_{n \times n,k}({\bf w}) &= \left( \bigcap_{{\bf w} \in \cup_{i = 1}^{\omega(d)} J_i} S_{n \times n}({\bf w}) \right) \cap D^a_{n \times n,k}\\
&= S_{n \times n}(\cup_{i = 1}^{\omega(d)} J_i) \cap D^a_{n \times n,k}\\
&= S^a_{n \times n,k}(\cup_{i = 1}^{\omega(d)} J_i)
\end{align*}
Finally let $m = |\langle \cup_{i = 1}^{\omega(d)} J_i \rangle|$, then using Lemma ~\ref{lemma:symmetric-config-2d-active-cell-intersect-size},
\begin{align*}
|S^a_{n \times n,k}(\cup_{i = 1}^{\omega(d)} J_i)| &= \binom{\frac{n^2}{m}}{\frac{k}{m}} (|\Sigma| - 1)^{\frac{n^2 - k}{m}},
\end{align*}
where $m = |\langle \cup_{i = 1}^{\omega(d)} J_i \rangle| = \prod_{i=1}^{\omega(d)} r_i^{\min(|J_i|,2)}$. \hfill $\Box$
\end{proof}


\begin{lemma-w-back-ref}[lemma:symmetric-config-2d-active-cell-size-alternative]
Pick $n,k \in {\mathbb{N}}$ with $k \leq n$ and let $d={\rm{gcd}}(k,n)$.
Let $n =\prod_{i=1}^{\omega(n)} p_i^{\alpha_i}$, $k=\prod_{i=1}^{\omega(k)} q_i^{\beta_i}$, and $d = \prod_{i=1}^{\omega(d)} r_i^{\gamma_i}$ be the prime factorizations of $n$, $k$, $d$, respectively. Then for any $a \in \Sigma$,
\begin{align*}
|S^a_{n \times n, k}| = \sum_{\substack{
{\bf 0} \lhd {\bf v} \unlhd {\bf 2} \\
{\bf v} \unlhd {\bf u} \unlhd {\rm top}({\bf v})}}
(-1)^{1 + |{\bf u}|}
\binom{\frac{n^2}{h({\bf v})}}{\frac{k}{h({\bf v})}}
(|\Sigma|-1)^{\frac{n^2-k}{h({\bf v})}}\\
\times \prod_{i=1}^{\omega(d)}\binom{r_i + 1}{u_i}, 
\end{align*}
where $h({\bf v}) = \prod_{i=1}^{\omega(d)} r_i^{\min(v_i,2)} \,$ and $\, {\rm top}({\bf v}) \in {\mathbb{Z}}^{\omega(d)}$ has $i$th coordinate $${\rm top}(i) = \begin{cases} v_i &\mbox{if } v_i < 2 \\
r_i + 1 & \mbox{if } v_i = 2. \end{cases}$$
\end{lemma-w-back-ref}
\begin{proof}
Similar to the proof of Lemma \ref{lemma:symmetric-config-overall-2d-size-alternative}.
\hfill $\Box$
\end{proof}


\begin{theorem-w-back-ref}[theorem:symmetric-config-2d-active-cell-size-final]
Pick $n,k \in {\mathbb{N}}$ with $k \leq n$ and let $d={\rm{gcd}}(k,n)$.
Let $n =\prod_{i=1}^{\omega(n)} p_i^{\alpha_i}$, $k=\prod_{i=1}^{\omega(k)} q_i^{\beta_i}$, and $d = \prod_{i=1}^{\omega(d)} r_i^{\gamma_i}$ be the prime factorizations of $n$, $k$, $d$, respectively. Then for any $a \in \Sigma$,
\begin{align*}
|S^a_{n \times n, k}| = \sum_{{\bf 0} \lhd {\bf v} \unlhd {\bf 2}}
(-1)^{1 + |{\bf v}|}
\binom{\frac{n^2}{h({\bf v})}}{\frac{k}{h({\bf v})}}
(|\Sigma|-1)^{\frac{n^2-k}{h({\bf v})}} \prod_{i=1}^{\omega(d)}r(i), 
\end{align*}
where $h({\bf v}) = \prod_{i=1}^{\omega(d)} r_i^{\min(v_i,2)} \,$ and $$ r(i) =
\begin{cases} 
    1 &\mbox{if } v_i = 0 \\
    p_i + 1 & \mbox{if } v_i = 1 \\
    p_i & \mbox{if } v_i = 2.
\end{cases}$$
\end{theorem-w-back-ref}
\begin{proof}
Similar to the proof of Theorem \ref{theorem:symmetric-config-overall-2d-size-final}. \hfill $\Box$
\end{proof}


\setcounter{theorem}{32}

\begin{corollary}
\label{cor:symmetric-config-2d-active-cell-size-2}
The number of binary symmetric configurations ($|\Sigma| = 2$)  with $k$ sites in state $a$ is given by
\begin{align*}
|S^a_{n \times n, k}| = \sum_{{\bf 0} \lhd {\bf v} \unlhd {\bf 2}}
(-1)^{1 + |{\bf v}|}
\binom{\frac{n^2}{h({\bf v})}}{\frac{k}{h({\bf v})}} \prod_{i=1}^{\omega(d)}r(i). \mbox{\hspace{40pt} $\Box$} 
\end{align*}
\end{corollary}


\begin{corollary}
\label{cor:symmetric-config-2d-active-cell-zero-size}
For any state set $\Sigma$ and state $a \in \Sigma$, the set $S^a_{n \times n,0}$ equals the set $S_{n \times n}$ for the state set $\Sigma \setminus \{ a \}$. \hfill $\Box$
\end{corollary}


\setcounter{theorem}{34}

\begin{theorem-w-back-ref}[theorem:detection-algorithm-worst-case-complexity]
The worst-case time complexity of the shift-symmetry detection algorithm for a square configuration of size $N = n^2$ is

$$O(n^3).$$

\end{theorem-w-back-ref}
\begin{proof}
In a worst-case scenario, when a configuration is non-shift-symmetric and there is only one cell breaking symmetry, each test requires to visit potentially all $n^2$ cells. The overall worst-case time complexity is therefore $O(|G_n| n^2)$. We know that the sum of distinct prime factors ${\rm sopf}(n) = \sum_{i = 1}^{\omega(n)} p_i$ also known as the integer logarithm is at most $n$ (if $n$ is prime), which gives us 

\begin{align*}
O(|G_n|n^2) &= O\left((\omega(n) + \sum_{i = 1}^{\omega(n)} p_i)n^2\right)\\
&= O((\log_2(n) + {\rm sopf}(n))n^2)\\
&= O(n^3).\mbox{\hspace{145pt} $\Box$}
\end{align*}
\end{proof}

\begin{theorem-w-back-ref}[theorem:detection-algorithm-average-case-complexity]
The average-case time complexity of the shift-symmetry detection algorithm for a square configuration of size $N = n^2$ generated from a uniform distribution is

$$O(n^2).$$

\end{theorem-w-back-ref}
\begin{proof}
Let $m = n^2 p^{-1}$ be the number of orbits for a prime $p$. Assuming a uniform distribution the probability of passing an orbit is $Q = |\Sigma|^{1 - p}$. If successful we move to a next orbit, otherwise we terminate with the probability $1 - Q$. The probability of terminating at $i$\textsuperscript{th} orbit can be therefore generalized as

$$ P_i = \begin{cases} 
    (1 - Q) Q^{i-1} &\mbox{if } i < m \\
    Q^{m - 1} & \mbox{if } i = m.
\end{cases}
$$

It is easy to show that these probabilities sum to $1$, i.e., we must terminate at one of $m$ orbits. Further, the probability of successfully passing the test for all the orbits---the probability that a configuration generated from a uniform distribution is shift-symmetric by a vector with an order $p$---equals $|\Sigma|^{n^2(p^{-1} - 1)}$.

By using the formula for a geometric sum we can prove that 


$$\sum_{i = 0}^{n - 1}(i + 1)r^i = \frac{1 - r^n(1 + n(1 - r))}{(1 - r)^2}.$$ 

We apply this to calculate the expected number of visited orbits as

\begin{align*}
E_p[\#orbits] &= \sum_{i = 1}^m iP_i \\
&= (1 - Q) \sum_{i = 0}^{m - 2}(i + 1)Q^i + mQ^{m - 1} \\
& = \frac{1 - Q^{m - 1}(m - Qm + Q) + (1 - Q)mQ^{m - 1}}{1 - Q}\\
& = \frac{1 - Q^{m}}{1 - Q}.
\end{align*}


Owing to $Q < 1$ we can bound the expected (average) number of visited orbits for a prime $p$ as

\begin{align*}
E_p[\#orbits] \leq (1 - Q)^{-1} = (1 - |\Sigma|^{1 - p})^{-1}.
\end{align*}

Each $p$-orbit contains $p$ cells and so the expected number of visited cells is simply

\begin{align*}
E_p[\#cells] \leq 2p(1 - |\Sigma|^{1 - p})^{-1}.
\end{align*}

Note that while moving from one orbit to a next one we can potentially revisit some cells, however, because the order is fixed we can visit each cell at most twice.

The number of generators $|G_n(p_i)|$ for each prime $p_i$ equals $p_i + 1$ (Eq. \ref{def:gen-sets}), thus the overall expected number of visited cells, i.e., the average-case time complexity in O-notation is

\begin{align*}
\sum_{i = 1}^{\omega(n)} (p_i + 1) p_i(1 - |\Sigma|^{1 - p_i})^{-1}.
\end{align*}

Since the expression $(1 - |\Sigma|^{1 - p_i})^{-1}$ is at most $2$ ($p_i \geq 2$) and the integer logarithm ${\rm sopf}(n)$ is at most $n$, the average-case time complexity of the shift-symmetry test is

\begin{align*}
O(\sum_{i = 1}^{\omega(n)} p_i^2 + \sum_{i = 1}^{\omega(n)} p_i) &= O({\rm sopf}^2(n) + {\rm sopf}(n)) \\
&= O(n^2).\mbox{\hspace{110pt} $\Box$}
\end{align*}
\end{proof}

%% file: 9-Appendix.tex
\section{Examples}
\label{appendix:examples}

\begin{example}
Let $n = 2^{\alpha_1}3^{\alpha_2}$, then using counting from Lemma
\ref{lemma:symmetric-config-overall-2d-size}, $|S_{n \times n}| = $ 
\begin{align*}
& \binom{3}{1}|\Sigma|^{\frac{n^2}{2}} + \binom{4}{1}|\Sigma|^{\frac{n^2}{3}}\\
- &\binom{3}{2}|\Sigma|^{\frac{n^2}{2^2}} - \binom{3}{1}\binom{4}{1}|\Sigma|^{\frac{n^2}{2\,3}} - \binom{4}{2}|\Sigma|^{\frac{n^2}{3^2}}\\
+ & \binom{3}{3}|\Sigma|^{\frac{n^2}{2^2}} + \binom{3}{2}\binom{4}{1}|\Sigma|^{\frac{n^2}{2^2\,3}} + \binom{3}{1}\binom{4}{2}|\Sigma|^{\frac{n^2}{2\, 3^2}}\\
& \hspace{123pt} + \binom{4}{3}|\Sigma|^{\frac{n^2}{3^2}}\\
- & \binom{3}{3}\binom{4}{1}|\Sigma|^{\frac{n^2}{2^2\,3}} - \binom{3}{2}\binom{4}{2}|\Sigma|^{\frac{n^2}{2^2\, 3^2}} - \binom{3}{1}\binom{4}{3}|\Sigma|^{\frac{n^2}{2\, 3^2}}\\
& \hspace{150pt} - \binom{4}{4}|\Sigma|^{\frac{n^2}{3^2}}\\
+ & \binom{3}{3}\binom{4}{2}|\Sigma|^{\frac{n^2}{2^2\,3^2}} + \binom{3}{2}\binom{4}{3}|\Sigma|^{\frac{n^2}{2^2\, 3^2}} + \binom{3}{1}\binom{4}{4}|\Sigma|^{\frac{n^2}{2\, 3^2}}\\
- &\binom{3}{3}\binom{4}{3}|\Sigma|^{\frac{n^2}{2^2\,3^2}} - \binom{3}{2}\binom{4}{4}|\Sigma|^{\frac{n^2}{2^2\, 3^2}}\\
+ & \binom{3}{3}\binom{4}{4}|\Sigma|^{\frac{n^2}{2^2\,3^2}}
\end{align*}
by Lemma \ref{lemma:symmetric-config-overall-2d-size-alternative}, $|S_{n \times n}| = $\\
\begin{align*}
& |\Sigma|^{\frac{n^2}{2}}\biggl[+\binom{3}{1}\biggl] + \\
& |\Sigma|^{\frac{n^2}{3}}\biggl[+\binom{4}{1}\biggl] + \\
& |\Sigma|^{\frac{n^2}{2\,3}}\biggl[-\binom{3}{1}\binom{4}{1}\biggl] + \\
& |\Sigma|^{\frac{n^2}{2^2}}\biggl[- \binom{3}{2}+ \binom{3}{3}\biggl] + \\
& |\Sigma|^{\frac{n^2}{3^2}}\biggl[- \binom{4}{2} + \binom{4}{3} - \binom{4}{4} \biggl] +\\
& |\Sigma|^{\frac{n^2}{2^2\,3}}\biggl[+ \binom{3}{2}\binom{4}{1} - \binom{3}{3}\binom{4}{1} \biggl] +\\
& |\Sigma|^{\frac{n^2}{2\, 3^2}}\biggl[+ \binom{3}{1}\binom{4}{2} - \binom{3}{1}\binom{4}{3} + \binom{3}{1}\binom{4}{4} \biggl] +\\
& |\Sigma|^{\frac{n^2}{2^2\,3^2}}\biggl[ - \binom{3}{2}\binom{4}{2} + \binom{3}{3}\binom{4}{2} + \binom{3}{2}\binom{4}{3}\\
& \hspace{34pt} - \binom{3}{3}\binom{4}{3} - \binom{3}{2}\binom{4}{4} + \binom{3}{3}\binom{4}{4} \biggl]
\end{align*}
and finally by Theorem \ref{theorem:symmetric-config-overall-2d-size-final}, $|S_{n \times n}| = $\\
\begin{align*}
& |\Sigma|^{\frac{n^2}{2}}3 + |\Sigma|^{\frac{n^2}{3}}4 - |\Sigma|^{\frac{n^2}{2\,3}}3 \cdot 4 - |\Sigma|^{\frac{n^2}{2^2}}2 - |\Sigma|^{\frac{n^2}{3^2}} 3 +\\
& |\Sigma|^{\frac{n^2}{2^2\,3}}2 \cdot 4  + |\Sigma|^{\frac{n^2}{2\, 3^2}}3 \cdot 3- |\Sigma|^{\frac{n^2}{2^2\,3^2}}2 \cdot 3
\end{align*}
\end{example}

\begin{example}
Let $n = 2^{\alpha_1}3^{\alpha_2}$, $a \in \Sigma$, and $k = 2^{\beta_1}3^{\beta_2}$, where $\beta_1 \leq \alpha_1, \beta_2 \leq \alpha_2$, and $\sigma = |\Sigma| - 1$. Then using counting from Lemma \ref{lemma:symmetric-config-2d-active-cell-size}  $|S^a_{n \times n,k}| = $
\begin{small}
\begin{align*}
& \binom{3}{1}\binom{\frac{n^2}{2}}{\frac{k}{2}}\sigma^{\frac{n^2 - k}{2}} + \binom{4}{1}\binom{\frac{n^2}{3}}{\frac{k}{3}}\sigma^{\frac{n^2 - k}{3}}\\
- &\binom{3}{2}\binom{\frac{n^2}{2^2}}{\frac{k}{2^2}}\sigma^{\frac{n^2 - k}{2^2}} - \binom{3}{1}\binom{4}{1}\binom{\frac{n^2}{2\,3}}{\frac{k}{2\,3}}\sigma^{\frac{n^2 - k}{2\,3}} - \binom{4}{2}\binom{\frac{n^2}{3^2}}{\frac{k}{3^2}}\sigma^{\frac{n^2 - k}{3^2}}\\
+ & \binom{3}{3}\binom{\frac{n^2}{2^2}}{\frac{k}{2^2}}\sigma^{\frac{n^2 - k}{2^2}} + \binom{3}{2}\binom{4}{1}\binom{\frac{n^2}{2^2 \, 3}}{\frac{k}{2^2 \, 3}}\sigma^{\frac{n^2 - k}{2^2\,3}}\\
& + \binom{3}{1}\binom{4}{2}\binom{\frac{n^2}{2 \, 3^2}}{\frac{k}{2 \, 3^2}}\sigma^{\frac{n^2 - k}{2\, 3^2}} + \binom{4}{3}\binom{\frac{n^2}{3^2}}{\frac{k}{3^2}}\sigma^{\frac{n^2 - k}{3^2}}\\
- & \binom{3}{3}\binom{4}{1}\binom{\frac{n^2}{2^2 \, 3}}{\frac{k}{2^2 \, 3}}\sigma^{\frac{n^2 - k}{2^2\,3}} - \binom{3}{2}\binom{4}{2}\binom{\frac{n^2}{2^2 \, 3^2}}{\frac{k}{2^2 \, 3^2}}\sigma^{\frac{n^2 - k}{2^2\, 3^2}}\\
& - \binom{3}{1}\binom{4}{3}\binom{\frac{n^2}{2 \, 3^2}}{\frac{k}{2 \, 3^2}}\sigma^{\frac{n^2 - k}{2\, 3^2}} - \binom{4}{4}\binom{\frac{n^2}{3^2}}{\frac{k}{3^2}}\sigma^{\frac{n^2 - k}{3^2}}\\
+ & \binom{3}{3}\binom{4}{2}\binom{\frac{n^2}{2^2 \, 3^2}}{\frac{k}{2^2 \, 3^2}}\sigma^{\frac{n^2 - k}{2^2\,3^2}} + \binom{3}{2}\binom{4}{3}\binom{\frac{n^2}{2^2 \, 3^2}}{\frac{k}{2^2 \, 3^2}}\sigma^{\frac{n^2 - k}{2^2\, 3^2}}\\
& + \binom{3}{1}\binom{4}{4}\binom{\frac{n^2}{2 \, 3^2}}{\frac{k}{2 \, 3^2}}\sigma^{\frac{n^2 - k}{2\, 3^2}}\\
- &\binom{3}{3}\binom{4}{3}\binom{\frac{n^2}{2^2 \, 3^2}}{\frac{k}{2^2 \, 3^2}}\sigma^{\frac{n^2 - k}{2^2\,3^2}} - \binom{3}{2}\binom{4}{4}\binom{\frac{n^2}{2^2 \, 3^2}}{\frac{k}{2^2 \, 3^2}}\sigma^{\frac{n^2 - k}{2^2\, 3^2}}\\
+ & \binom{3}{3}\binom{4}{4}\binom{\frac{n^2}{2^2 \, 3^2}}{\frac{k}{2^2 \, 3^2}}\sigma^{\frac{n^2 - k}{2^2\,3^2}}
\end{align*}
\end{small}

by Lemma \ref{lemma:symmetric-config-2d-active-cell-size-alternative} $|S^a_{n \times n,k}| = $
\begin{small}
\begin{align*}
& \binom{\frac{n^2}{2}}{\frac{k}{2}}\sigma^{\frac{n^2-k}{2}}\biggl[
+\binom{3}{1}
\biggl] + \\
& \binom{\frac{n^2}{3}}{\frac{k}{3}}\sigma^{\frac{n^2-k}{3}}\biggl[
+\binom{4}{1}
\biggl] +\\
& \binom{\frac{n^2}{2 \, 3}}{\frac{k}{2 \, 3}}\sigma^{\frac{n^2-k}{2\,3}}\biggl[
-\binom{3}{1}\binom{4}{1}
\biggl] + \\
& \binom{\frac{n^2}{2^2}}{\frac{k}{2^2}}\sigma^{\frac{n^2-k}{2^2}}\biggl[
- \binom{3}{2}
+ \binom{3}{3}
\biggl] +\\
& \binom{\frac{n^2}{3^2}}{\frac{k}{3^2}}\sigma^{\frac{n^2-k}{3^2}}\biggl[
- \binom{4}{2}
+ \binom{4}{3}
- \binom{4}{4}
\biggl] +\\
& \binom{\frac{n^2}{2^2 \, 3}}{\frac{k}{2^2 \, 3}}\sigma^{\frac{n^2-k}{2^2\,3}}\biggl[
+ \binom{3}{2}\binom{4}{1}
- \binom{3}{3}\binom{4}{1}
\biggl] +\\
& \binom{\frac{n^2}{2 \, 3^2}}{\frac{k}{2 \, 3^2}}\sigma^{\frac{n^2-k}{2\, 3^2}}\biggl[
+ \binom{3}{1}\binom{4}{2}
- \binom{3}{1}\binom{4}{3}
+ \binom{3}{1}\binom{4}{4}
\biggl] +\\
& \binom{\frac{n^2}{2^2 \, 3^2}}{\frac{k}{2^2 \, 3^2}}\sigma^{\frac{n^2-k}{2^2\,3^2}}\biggl[
- \binom{3}{2}\binom{4}{2}
+ \binom{3}{3}\binom{4}{2}
+ \binom{3}{2}\binom{4}{3}\\
& \hspace{85pt} - \binom{3}{3}\binom{4}{3}
- \binom{3}{2}\binom{4}{4}
+ \binom{3}{3}\binom{4}{4} \biggl]
\end{align*}
\end{small}

and finally by Theorem \ref{theorem:symmetric-config-2d-active-cell-size-final}, $|S^a_{n \times n,k}| = $
\begin{align*}
& \binom{\frac{n^2}{2}}{\frac{k}{2}}\sigma^{\frac{n^2-k}{2}}3 + \binom{\frac{n^2}{3}}{\frac{k}{3}}\sigma^{\frac{n^2-k}{3}}4 - \binom{\frac{n^2}{2 \, 3}}{\frac{k}{2 \, 3}}\sigma^{\frac{n^2-k}{2\,3}}3 \cdot 4\\
-& \binom{\frac{n^2}{2^2}}{\frac{k}{2^2}}\sigma^{\frac{n^2-k}{2^2}}2 - \binom{\frac{n^2}{3^2}}{\frac{k}{3^2}}\sigma^{\frac{n^2-k}{3^2}}3 + \binom{\frac{n^2}{2^2 \, 3}}{\frac{k}{2^2 \, 3}}\sigma^{\frac{n^2-k}{2^2\,3}}2 \cdot 4\\
+& \binom{\frac{n^2}{2 \, 3^2}}{\frac{k}{2 \, 3^2}}\sigma^{\frac{n^2-k}{2\, 3^2}}3 \cdot 3 - \binom{\frac{n^2}{2^2 \, 3^2}}{\frac{k}{2^2 \, 3^2}}\sigma^{\frac{n^2-k}{2^2\,3^2}}2 \cdot 3
\end{align*}
\end{example}